\documentclass[aps,10pt,prx,nofootinbib,notitlepage]{revtex4-2}
\usepackage{header}

\begin{document}
\title{Locality and filter design for dissipative ground-state preparation}

\author{Samuel~J.~Elman}
\affiliation{Centre for Quantum Software and Information, School of Computer Science, Faculty of Engineering and Information Technology, University of Technology Sydney, NSW 2007, Australia}

\date{\today}

\begin{abstract}
    Ground-state preparation is a fundamental challenge in quantum many-body physics, chemistry and materials science. Dissipative state-preparation techniques engineer a Lindbladian coupling to a bath that drives the system to its ground state, needing no initial overlap. Their cost is dominated by the Hamiltonian simulations inside the filtered jump operators, since each jump is a weighted sum of Heisenberg-evolved copies, or branches, of a coupling operator. We reduce this cost by following three simple observations to their logical end. First, on a fixed time grid, linear programming optimises the branch weights, trading the heating response against the circuit cost. Second, Lieb--Robinson bounds let each branch evolve on a patch whose radius grows with its duration. Third, if the dynamics converges quickly and truncation changes the jumps only slightly, the stationary state stays near the ground state even on patches smaller than the light cone.
    We validate our approach numerically. On a grid of $81$ evolution times, the optimised filter reduces the worst-case heating response, or leakage, more than $116$-fold at equal cost. Across fifteen molecules and six chains, the cost falls by 27 to 56\% at matched design-grid leakage; at matched cost, the molecular leakage falls by a factor of $2.9$ to $2\times10^{3}$. Density-matrix simulations show high ground-state fidelity for both filters, and on Ising and Hubbard chains the cheaper filter keeps a ground-state population of at least $0.89$ on patches of three or four sites. The same patches fail on a Heisenberg chain whose truncation error falls more slowly with patch size and whose Lindbladian gap is 12 times smaller than the Ising chain's, pointing to a role for the gap in how far patches can be truncated. At fixed radius the per-branch cost becomes independent of system size, reducing the cost of dissipative state preparation.
\end{abstract}
\maketitle

\section{Introduction}
\label{sec:intro}

Preparing the ground state of a many-body quantum Hamiltonian is among the oldest and most consequential targets of quantum computation~\cite{AbramsLloyd1999,PoulinWocjan2009,GeTuraCirac2019}.
Ground-state preparation is the entry point to many quantum-chemistry and materials algorithms, where quantum computers are expected to provide the most utility~\cite{Lee2023evidence}, and it is the core primitive of combinatorial optimisation recast as energy minimisation~\cite{FarhiAdiabatic2000}.
Ground-state preparation is also where alleged quantum advantage may be most fragile: in the worst case, even finding the ground energy of a local~\cite{KitaevShenVyalyi2002,KempeKitaevRegev2006} or electronic structure~\cite{SchuchVerstraete2009,OGorman2022} Hamiltonian is hard for the class QMA. Despite this, much effort has gone into algorithms for the physically motivated instances in which one either possesses a good ansatz for the ground state~\cite{LinTong2020,DongLinTong2022} or can afford to cool the system slowly~\cite{FarhiAdiabatic2000,AlbashLidar2018}.

Within the fault-tolerant literature there are, broadly speaking, two classes of ground-state preparation algorithms. Coherent approaches, such as discrete adiabatic simulation~\cite{FarhiAdiabatic2000,AlbashLidar2018,AnCostaBerry2025,Costa2022discrete}, quantum phase estimation~\cite{Kitaev1995,AbramsLloyd1999,Babbush2018}, and eigenspace filtering, whether by polynomial transformations of a block encoding~\cite{LinTong2020,LinTong2020filter} or of the time-evolution unitary~\cite{DongLinTong2022,LinTong2022heisenberg}, require a promise on the initial state such as a non-negligible overlap with the target, or an efficient preparation along a gapped path to the desired state, neither of which can be guaranteed for strongly correlated systems~\cite{Lee2023evidence}.

Dissipative approaches engineer open-system dynamics whose unique fixed point is the ground state~\cite{Diehl2008,Verstraete2009}, a class of dynamics that a quantum computer can itself simulate efficiently~\cite{Kliesch2011}. The concept derives from quantum Gibbs sampling, which has a long history in quantum thermodynamics~\cite{Temme2011,KastoryanoBrandao2016} and which recent algorithmic work~\cite{RallWangWocjan2023,WocjanTemme2023} has brought to the fore. Chen, Kastoryano, Brand\~{a}o and Gily\'{e}n constructed Lindbladians that prepare thermal states under approximate detailed balance~\cite{Chen2023thermal}, made exactly detailed balanced for arbitrary noncommuting Hamiltonians by Chen, Kastoryano and Gily\'{e}n~\cite{Chen2023gibbs}, with subsequent generalisations of the Glauber and Metropolis dynamics~\cite{Gilyen2024glauber} and of the Kubo--Martin--Schwinger condition~\cite{DingLiLin2025kms}; Cubitt gave a dissipative eigensolver whose fixed point is the ground state of an arbitrary Hamiltonian~\cite{Cubitt2023dqe}; and Ding, Chen and Lin constructed a single-ancilla dynamics for which the ground state of a general gapped Hamiltonian is a fixed point and which moves population into the ground space even from a zero-overlap input~\cite{DingChenLin2024}, with the underlying Lindblad dynamics realisable by general-purpose simulation methods~\cite{DingLiLin2024}.

Reaching the ground state from every initial state takes more: the couplings must, between them, connect every state to the ground state, or the rate of convergence must be bounded separately. For Hamiltonians that conserve the number of excitations, with couplings that only remove excitations, a condition on the interaction graph known as zero forcing is enough, and the dissipation can then act on a single site~\cite{BeerBurgarth2026}. These developments bypass the central limitation of phase estimation, the need for a non-negligible initial overlap~\cite{DingChenLin2024}, and admit end-to-end efficiency guarantees~\cite{DingZhanPreskillLin2025}. Reducing the ancilla register to a single qubit makes the protocol appear particularly attractive for early fault-tolerant algorithms, where logical width is the scarcer resource.

Theoretical advances are also identifying a growing class of systems for which these algorithms ought to be efficient: the mixing time is logarithmic in the system size for weakly interacting spin and fermionic models~\cite{Zhan2025}, rapid thermalisation has been proved for commuting spin chains at any temperature~\cite{Bardet2023} and for general local Hamiltonians above a threshold temperature~\cite{RouzeFrancaAlhambra2024}, and finite-patch versions of the Gibbs sampler mix rapidly at high temperature~\cite{Hahn2026}. Rapid mixing does not by itself imply a quantum advantage, since the same weakly interacting regime is where cluster expansions give classical algorithms, at low temperature~\cite{HelmuthMann2023}, at high temperature~\cite{MannMinko2024} and, more recently, at arbitrary temperature~\cite{WaiteMann2026}; the regimes of interest are those where a fast-mixing dissipative dynamics exists and no such expansion converges. The example systems in this paper are deliberately classically tractable, so that the dissipative dynamics can be checked exactly. The construction is now being carried into applications: to excited states, by reformulating them as effective ground states~\cite{LiLin2025excited}; to chemically difficult transition-state geometries, by continuing the ground state dissipatively along a reaction coordinate~\cite{Watts2026continuation}; and to deciding which phase a system is in, by reading phase-sensitive observables from the early-time cooling dynamics without preparing the ground state at all~\cite{LiYangLin2026phase}.

The central object of all of these dissipative constructions is a set of \emph{filtered jump operators}: a local coupling operator $A_a$, which drives transitions between Hamiltonian eigenstates, dressed with an energy filter, so that the jump lowers the energy and does not raise it. This filter is realised via Fourier decomposition: evolving the coupling under the Hamiltonian $H$ for a set of times $(t_j)_j$ and combining the results, which we call branches, with weights $(c_j)_j$. These time evolutions account for most of the gate cost of the dissipative primitive. Every evolution is a full Hamiltonian simulation, and the filter of Ref.~\cite{DingChenLin2024}, which we call the Gevrey filter after the function class its convergence proof assumes, calls for tens to thousands of them per jump. The single-ancilla construction saves width, meaning logical qubits, and pays for it in gate count. When compiled under block-encoding conventions similar to those of the coherent approaches, such as Refs.~\cite{Babbush2018,LowChuang2019,Harrigan2024}, the costs land many orders of magnitude above the coherent estimates for comparable systems~\cite{KanSymons2025,Lee2021}. The narrow register is bought at a price in non-Clifford gates that is difficult to justify for near-term fault-tolerance.

In this paper, we present a construction that reduces the cost of these filtered jumps. In particular, we price one block encoding of one filtered jump; the costs of the outer loop that simulates the Lindblad dynamics, the dissipative channel, its splitting across jumps, ancilla resets and the stopping measurement are excluded. The jumps constructed here have finite range, so the resulting Lindbladian falls within the scope of locality-exploiting simulation algorithms~\cite{Mizuta2026}, whose \emph{patching} decomposes the dissipative evolution itself into subsystem evolutions and is distinct from the patches used here, which bound the support of each branch inside a single jump. Under our cost model (\cref{app:frontier}), for the transverse-field Ising model (TFIM) on a square lattice at a fixed normalised gap, compiled with the shared-evolution circuit of \cref{sec:theory}, the Gevrey filter~\cite{DingChenLin2024} costs $\PerJumpDenseeight$ Toffolis per jump at eight by eight sites and $\PerJumpDensehundred$ at a hundred by a hundred, while the jumps constructed here, in which each evolution acts on a \emph{patch}, a ball of sites around the coupling rather than the whole lattice, cost $\PerJumpMatchThree$ regardless of the lattice size. The gains behind those numbers are set out in \cref{sec:numerics}, model by model.

Our construction is based on three observations, each of which changes one part of the current filtered-jump protocol.
First, a cheaper filter function may be obtained at the same \emph{leakage}, the filter's largest response to energy-raising transitions, by optimising the Fourier coefficients directly against the leakage and the circuit cost, in place of designing a continuous window and approximating it. Previous filters~\cite{DingChenLin2024,Watts2026continuation,Bothe2026Early} have been designed in the frequency domain and then approximated in time, so they suppress heating less than their maximum evolution time allows. Optimising in the time domain gives up a frequency response that can be written down in closed form, and returns a filter that is immediately implementable. On a fixed symmetric time grid, the optimum over the family of reweighted filters is global.

The second observation is a quantitative locality bound. For a geometrically local Hamiltonian of bounded interaction strength, the Lieb--Robinson bound~\cite{LiebRobinson1972,NachtergaeleSims2006,HastingsKoma2006} confines a Heisenberg-evolved jump operator to an effective light cone. A coupling operator evolved for a given time under a local Hamiltonian is affected only by the Hamiltonian terms within a distance set by that time and the Lieb--Robinson velocity, up to exponentially small corrections. Each branch of the filter can therefore be generated by the Hamiltonian restricted to a ball around the jump's centre, whose radius is chosen for that branch's own time. Because the light cone is set by the elapsed time, a circuit that builds the branches from shared evolution pieces can give each piece the radius that its own elapsed time requires, so the early pieces run on small patches, and the finer the pieces, the more of the light cone they cut away. For the circuit developed here, this grading keeps the same error bound as giving every branch its own radius.

The final observation is that the stationary state of the Lindblad dynamics is not displaced too dramatically when the patches are truncated at a maximum radius far below the light cone of the long branches. Adapting the neighbourhood construction of~\citet{Hahn2026} to a one-sided filter, which keeps energy-lowering transitions and suppresses energy-raising ones, and pairing it with an elementary perturbation bound for the Lindblad evolution, we show that if the Lindbladian with untruncated jumps has a known mixing rate and prefactor (a mixing pair, \cref{sec:theory}), the truncated dynamics mixes and its stationary state lies within a computable distance of the ground state, controlled by the leakage of the filter and by how much the truncation changes the Lindbladian. Stability of local dissipative dynamics under perturbation has a longer history on the lattice, where it is established for local observables under a rapid-mixing hypothesis~\cite{Cubitt2015}; the estimate used here is global, in trace norm, and needs only the mixing pair. The bound is a worst case that grows with the number of jumps, and it certifies a maximum radius below the light cone of the longest branch only when the truncated branches carry little coefficient weight. Outside that regime the evidence is numerical: density-matrix simulations of transverse-field Ising, Heisenberg and Hubbard chains of at most eight sites show the actual error far below the bound in the rapidly mixing instances, with patches of three to seven sites, and show the same patches failing on a slowly mixing Heisenberg chain.

Thus our result may be summarised as follows: the maximum evolution time is set by the gap, the radius by the time, and cost-aware jumps on small overlapping patches cool cheaply where the dynamics mixes quickly.

The rest of this manuscript is organised as follows. \Cref{sec:theory} states the construction and the results that support it, with proofs deferred to the appendices: the two-stage optimisation of the filter in the time domain and the limit on how far the coefficients move from the reference, which keeps the dynamics convergent; the branch-adaptive locality bound, with a compilation of the jump that shares evolutions between branches and the radius schedule it admits; and the conditional stationary-state bound for the dynamics truncated at a maximum radius. \Cref{sec:numerics} reports the numerical evidence: filter designs and cost reductions for the molecules on which the Gevrey filter has been used~\cite{DingChenLin2024,LiZhanLin2025ab}, Lindblad simulations of the chains, and the compiled cost of the design optimisations acting together, under the fault-tolerant cost model of \cref{app:frontier}, which follows the block-encoding conventions of Refs.~\cite{Babbush2018,LowChuang2019,Harrigan2024} and which also costs the two-dimensional TFIM. \Cref{sec:outro} summarises our findings, places them in context and outlines future work.

\section{Construction and guarantees}
\label{sec:theory}

\emph{Setting.}
A Hamiltonian $H = \sum_{X \in E} h_X$ is the Hermitian sum of weighted terms on an interaction hypergraph $(\mathcal{V}, E)$ with $n$ qubit vertices. Each hyperedge has at most $k$ vertices, each vertex is a member of at most $\Delta_G$ hyperedges, and the norm of each term is bounded, $\norm{h_X} \le J$, where $\norm{\cdot}$ is the operator norm.

The hypergraph carries the shortest-path metric: $\dist(v,v')$ is the smallest number of steps in a sequence of vertices from $v$ to $v'$ in which each consecutive pair belongs to a common hyperedge, and $\dist(v,v')=\infty$ when no such sequence exists.
The ball of radius $R$ about a vertex $a\in \mc{V}$ is $B_R(a)\coloneqq\{v\mid \dist (v,a)\le R\}$.
Let 
\begin{equation}
    V(R) \coloneqq \max_v |B_R(v)|
\end{equation}
denote the \textit{single-vertex growth function}, which gives the maximum number of vertices in a radius-$R$ ball. On a $D$-dimensional hypercubic lattice with interaction range $r_0$, as in the numerical examples, the growth function satisfies $V(R)\leq (2r_0R+1)^D$.

For a subset of vertices $Y\subseteq \mc{V}$, define the radius-$R$ ball about $Y$ by
\begin{equation}
    B_R(Y)\coloneqq\bigcup_{y\in Y}\{v\mid \dist(v,y)\leq R\}.
\end{equation}
Since this is the union of $|Y|$ single-vertex balls, its size obeys $|B_R(Y)|\leq |Y|\,V(R)$. When $Y$ is itself a ball of radius $r_S$ about a vertex $a$, we write $S_a=B_{r_S}(a)$, and call a ball $B_R(S_a)$ about it a \textit{patch}.
We denote the restriction of the Hamiltonian to terms whose support lies \textit{within} a region $Y\subseteq \mc{V}$ as $H_Y \coloneqq \sum_{X \subseteq Y} h_X$.

The dynamics are defined relative to a \emph{working space}, a symmetry sector selected for the problem.
We define $\Pi_{\mathrm S}$ as the orthogonal projector onto this sector, so that for all states $\rho$ in the working space, $\rho=\Pi_{\mathrm S}\rho \Pi_{\mathrm S}$.
We require $[\Pi_{\mathrm S},H]=0$, so the Hamiltonian preserves the working space, and require each jump operator defined below to commute with $\Pi_{\mathrm S}$ as well.
Every norm below is taken on this space.
For a bounded operator $X$ on the working space, $\tnorm{X}=\Tr\sqrt{X^\dagger X}$ is the trace norm, and for a linear map $\mc{M}$ on such operators the induced trace norm is $\indnorm{\calM}=\sup_{\tnorm{X}\le1}\tnorm{\calM(X)}$.
The trace distance between states is half the trace norm of their difference.
For the molecules studied, the working space has fixed electron number and spin sector; for the Hubbard chain, it is the sector with two electrons of each spin, $N_\uparrow=N_\downarrow=2$; for the spin chains, it is the full Hilbert space.
On the working space, we assume that $H$ has a unique ground state, $\Pi_0=\ket{\psi_0}\bra{\psi_0}$, and write $H\ket{\psi_k}=E_k\ket{\psi_k}$ with $E_0<E_1\leq\cdots$. We write $\Delta_{\rm spec}=E_1-E_0$ for its spectral gap and use a design gap $0<\Delta\leq\Delta_{\rm spec}$.
Bounded degree, hyperedge size and term norms give $H$ a Lieb--Robinson bound~\cite{LiebRobinson1972,HastingsKoma2006,NachtergaeleSims2006}: as operators are Heisenberg evolved under $H$, their support spreads at a finite velocity $v_\LR$, of order $J k \Delta_G$, with tails outside the light cone decaying at a rate $\mu$, provided the growth function is subexponential, or exponential at a rate at most $\mu/2$.

\emph{Lindblad dynamics.}
The main object of this work is the Lindblad master equation
\begin{equation}
    \dot\rho=\mc{L}[\rho]=\underbrace{-i[H,\rho]}_{\mc{L}_H[\rho]}+\sum_{a=1}^{m}\underbrace{\Bigl(K_a\rho K_a^\dagger -\tfrac{1}{2}\{K_a^\dagger K_a,\rho\}\Bigr)}_{\mc{D}[K_a](\rho)},
    \label{ea:lind}
\end{equation}
where $\mc{L}_H$ describes the coherent evolution, which preserves the energy, and $\mc{D}[K_a]$ is the dissipator of the jump operator $K_a$, one of the $m$ jumps to be designed.
The jumps couple the energy eigenstates, transferring population between them, and can either lower or raise the energy.
In particular, starting from an energy eigenstate $\ket{\psi_{k'}}$, the rate of change of energy due to the jumps is
\begin{equation}
    \left.\frac{d}{dt}\Tr(H\rho)\right|_{\rho=\ket{\psi_{k'}}\bra{\psi_{k'}}}
    =\sum_{a=1}^{m}\sum_k
    (E_k-E_{k'})\bigl|\bra{\psi_k}K_a\ket{\psi_{k'}}\bigr|^2.
\end{equation}
A transition with $E_k-E_{k'}<0$ lowers the energy; one with $E_k-E_{k'}>0$ raises it.
We therefore filter the jumps so that they retain the transitions that cool and suppress those that add energy.

We say that $\mc{L}$ \emph{mixes} with \emph{mixing pair} $(\lambda,\kappa)$, where $\lambda>0$ and $\kappa\ge1$, if $\tnorm{e^{\mc{L}t}X}\le\kappa e^{-\lambda t}\tnorm{X}$ for every traceless Hermitian $X$ and $t\ge0$. It then has a unique stationary state $\sigma$, and every initial state is within trace distance $\varepsilon$ of $\sigma$ after time $\lambda^{-1}\log(\kappa/\varepsilon)$, which bounds the \emph{mixing time}. The dynamics mixes \emph{rapidly} when $\lambda$ does not shrink with the system size and $\kappa$ grows at most polynomially in it, so that the mixing time grows only logarithmically. The \emph{Lindbladian gap}, minus the largest real part of a nonzero eigenvalue of $\mc{L}$, bounds every mixing rate $\lambda$ from above.

\emph{Filtered jumps.}
Following Refs.~\cite{DingChenLin2024,DingZhanPreskillLin2025}, we choose $m$ localised couplings $A_1,\ldots,A_m$ with $\norm{A_a}\le1$ and $\supp(A_a)\subseteq S_a=B_{r_S}(a)$. We engineer the desired frequency selectivity by filtering each local coupling $A_a$, so that its response is suppressed at positive energy differences and large at negative ones.
For a local lattice model, $m=O(n)$; for fermionic molecular models, the couplings are hoppings between pairs of the $n_{\rm o}$ spatial orbitals, giving $m=O(n_{\rm o}^2)$. Each coupling commutes with $\Pi_{\mathrm S}$, so the filtered jumps below preserve the working space.

Each coupling contains components that shift population between energy eigenstates. Under Heisenberg evolution by the Hamiltonian $H$, a component connecting $\ket{\psi_{k'}}$ to $\ket{\psi_k}$ acquires the phase $e^{i(E_k-E_{k'})t}$. By taking a linear combination of the Heisenberg-evolved coupling at different times and with suitable weights, we therefore multiply each transition by a chosen function of its energy difference.
If the target response is $\hat h(\omega)=\int_{-\infty}^{\infty} f(t)e^{i\omega t}\,dt$, the corresponding filtered coupling is
\begin{equation}
    \int_{-\infty}^{\infty} f(t)\,e^{iHt}A_a e^{-iHt}\,dt\approx \sum_{j=-K}^{K} \tau f(t_j)\,e^{iHt_j}A_a e^{-iHt_j},
\end{equation}
where we have approximated the time integral by a finite sum on the grid $t_j=j\tau$, $|j|\le K$.
Each time-evolved operator in the sum gives one term in the filtered jump; we refer to the summands as \emph{branches}: branch $j$ evolves for time $t_j$, applies the coupling, and evolves back.
For $N=2K+1$ candidate branches, labelled $j=-K,\ldots,K$, with time nodes $\bs{t}=(t_j)$ and complex coefficients $\bs{c}=(c_j)$, the filtered jump and its frequency response are
\begin{equation}
    K_{a,N}(\bs{c})\coloneqq
    \sum_{j=-K}^{K}c_j e^{iHt_j}A_a e^{-iHt_j},
    \qquad
    h_N(\omega;\bs{c})\coloneqq\sum_{j=-K}^{K}c_j e^{i\omega t_j}.
    \label{eq:mjump}
\end{equation}
We seek cheaper instantiations of this filtered jump through careful analysis of its coefficients and times. In the energy eigenbasis,
\begin{equation}
    \bra{\psi_k}K_{a,N}(\bs{c})\ket{\psi_{k'}}
    =h_N(E_k-E_{k'};\bs{c})\bra{\psi_k}A_a\ket{\psi_{k'}}.
\end{equation}
Thus the desired response determines the amplitude of each transition induced by the local coupling.
We choose the weight of each branch to suppress upward transitions on the heating band $[\Delta,W]$ while maintaining a prescribed minimum cooling $P(e)>0$ at frequency $-e$, for $e\in[\Delta,W]$. The upper edge $W$ covers every transition that a coupling can drive out of the ground state; for the numerics below, we rescale $H$ so that its spectrum lies within $[-1,1]$, and $W=2$. The worst-case heating amplitude on the heating band is the \emph{heating leakage} $L_\Delta[h_N]\coloneqq\sup_{e\in[\Delta,W]}|h_N(e)|$.
The coefficient one-norm $\norm{\bs{c}}_1 \coloneqq \sum_j |c_j|$ bounds the jump's operator norm and sets the normalisation of its block encoding; $T \coloneqq \max_j |t_j|$ is the maximum evolution time. The engineered jumps then generate the Lindblad dynamics
\begin{equation}
    \calL_N[\rho] \coloneqq -i[H, \rho] + \sum_{a=1}^{m} \Bigl( K_{a,N} \rho K_{a,N}^\dagger - \tfrac{1}{2}\{ K_{a,N}^\dagger K_{a,N}, \rho \} \Bigr).
    \label{eq:lindblad_K}
\end{equation}

\emph{The reference filter.} The filter is the main design choice in a dissipative cooling algorithm. Although multiple ans\"atze for such filters have been proposed~\cite{Bothe2026Early,Watts2026continuation}, to the best of our knowledge, the cheapest previously proposed filter for the same leakage is the \textit{Gevrey class filter} defined in Ref.~\cite{DingChenLin2024}. We therefore use it as the reference, and optimise the implementation by reweighting its coefficients.

\citet{DingChenLin2024} start from a target response in the frequency domain, defined by the smooth \emph{window} function
\begin{equation}
    \hat{h}_{\eref}(\omega) \;=\; \tfrac{1}{2}\left[\erf\left(\tfrac{\omega + \omega_{\rm w}}{\delta_{\rm w}}\right) - \erf\left(\tfrac{\omega + \Delta}{\Delta}\right)\right],
    \qquad (\omega_{\rm w}, \delta_{\rm w}) = (2.5, 0.5),
    \label{eq:window}
\end{equation}
that is close to one in the cooling regime, at negative $\omega$, and close to zero in the heating regime, at positive $\omega$.
The time integral is then discretised with the trapezoidal rule on the grid $t_j=j\tau$, $|j|\le K$. The coefficients are chosen to be $d_j=\tau f(t_j)$, halving the two endpoint coefficients, where $f(t)$ is the continuous Fourier inverse of $\hat{h}_{\eref}(\omega)$, which has a closed form; we write $h_{\eref,N}(\omega)\coloneqq h_N(\omega;\bs d)$ for the resulting truncated response.
These \emph{Gevrey coefficients} obey the constraint $d_{-j} = d_j^{*}$.
In \cref{sec:numerics}, the molecular and chain prescription uses $T_\mrm{target}=5/\Delta$, $\tau=\pi/5$, $K=\operatorname{round}(T_{\mrm{target}}/\tau)$ and realised $T=K\tau$. 
\subsection{Fixed-Time Linear Program}

Our first contribution is the observation that, with the times fixed and the coefficients of $t_j$ and $-t_j$ paired, the filter design is a linear program, so its worst-case heating response can be minimised exactly.
Filters designed in the frequency domain and then approximated suffer uncontrolled leakage from the finite truncation, whereas designing in the time domain lets us minimise the leakage directly.
This construction allows two approaches to filter design: the first minimises the leakage at fixed cost; the second fixes the leakage and minimises the cost.

Branches with negligible reference coefficients, $d_j\approx0$, are removed first by fixing their optimised coefficients $c_j$ to zero; all sums and constraints below run over the remaining branches, while $N$ still counts every candidate label.
The \emph{paired-reweighting family} is
\begin{equation}
    c_j(\bs{x}) \coloneqq x_j d_j, \qquad x_{-j} = x_j \in \mbb{R},
    \qquad\text{such that}\qquad
    h_N(e; \bs{c}(\bs{x})) = x_0 d_0 + \sum_{\mrm{pairs}\ j>0} x_j\, \varphi_j(e),
    \label{eq:mfamily}
\end{equation}
where $\varphi_j(e) \coloneqq 2|d_j| \cos(e t_j + \arg d_j)$ for $j>0$, a real series for which the times $t_j$ are the fixed frequencies of the basis functions and the amplitudes $x_j$ are the variables, and the self-paired zero-time term $x_0 d_0$ ($d_0$ real) enters once. The displacement $\eta(\bs{c},\bs{d}) \coloneqq \sum_j |c_j - d_j|$ measures how far the coefficient vector has moved from the Gevrey construction~\cite{DingChenLin2024}, while the branch cost $W_j(R_j, t_j) \ge 0$ is the compiled cost of executing branch $j$ once on a patch of radius $R_j$; a binary $z_j \in \{0,1\}$ records whether branch $j$ is active, through $|x_j| \le M_j z_j$ with $M_j$ a valid upper bound on $|x_j|$; the active branches form the \emph{support} of the filter.

\emph{The filter as a linear program.}
A linear program minimises or maximises a linear function of real variables subject to linear equality and inequality constraints. A mixed-integer linear program is a linear program in which some variables must take integer values, such as the binary indicators $z_j\in\{0,1\}$ here. With the time nodes fixed, the paired response is real and affine in the amplitudes. On a finite sample, its largest absolute value is represented by an upper-bound variable constrained above both the response and its negative at every sample point. The cooling constraints are affine, and the coefficient one-norm becomes linear when each amplitude is split into its positive and negative parts. Thus the spectral stage below is a linear program.

\begin{proposition}[The two-stage program]
\label{thm:twostage}
    Fix a finite sample $\mathcal E \subset [\Delta, W]$, the design grid, and a coefficient budget $\Lambda_{\max}$, and define the sampled heating leakage $L_{\mathcal E}[h] \coloneqq \max_{e \in \mathcal E} |h(e)|$. Let $\mathcal{F}$, assumed nonempty, be the paired amplitudes $\bs{x}$ with $\norm{\bs{c}(\bs{x})}_1 \le \Lambda_{\max}$ and $h_N(-e; \bs{c}(\bs{x})) \ge P(e)$ on $\mathcal E$. The spectral stage
    \begin{equation}
        L^\star_{\mathcal E} \;\coloneqq\; \min_{\bs{x} \in \mathcal{F}}\; L_{\mathcal E}\bigl[ h_N(\cdot\,; \bs{c}(\bs{x})) \bigr]
        \label{eq:stage1}
    \end{equation}
    is a linear program with a global minimiser $\bs{c}_{\mathrm{spec}}$. For costs $W_j \ge 0$, an allowance $\xi \ge 0$ and any $\eta_{\sup} \ge \eta(\bs{c}_{\mathrm{spec}}, \bs{d})$, the support stage
    \begin{equation}
        C^{\star}_{\mathrm{br}}(\xi) \coloneqq \min_{\bs{x} \in \mathcal{F},\, \bs{z}}\; \sum_{j} W_j z_j
        \quad\text{subject to}\quad
        L_{\mathcal E}\bigl[ h_N(\cdot\,; \bs{c}(\bs{x})) \bigr] \le (1+\xi)\,L^\star_{\mathcal E},
        \quad
        \eta(\bs{c}(\bs{x}), \bs{d}) \le \eta_{\sup},
        \quad
        |x_j| \le M_j z_j,
    \label{eq:stage2}
    \end{equation}
    is a mixed-integer linear program, whose minimiser is the cheapest filter in the family meeting both allowances.
\end{proposition}
The proof is provided in \cref{app:obs1}. The least heating leakage any filter in the family can reach is the value of a linear program, and the cheapest filter that stays within a stated error on that leakage solves an integer program over which branches are kept. Treating the times as data allows us to find a global optimum; the support stage deletes branches without moving them.
Imposing the constraints on the whole band $[\Delta,W]$ gives infinitely many constraints, one for each energy. The second-derivative bounds in \cref{app:obs1} replace them by finitely many sufficient conditions; optimality for this whole-band problem requires a separate argument.
Under the ladder compilation of \cref{eq:costmain} used in \cref{sec:numerics}, and at a fixed value of its accounting multiplier $A_{\rm norm}$, the cost depends on the support only through its longest active branch, so the cheapest feasible support is found by enumerating truncations rather than by solving \cref{eq:stage2} as a general integer program.

\begin{figure*}[t]
    \includegraphics[width=\textwidth]{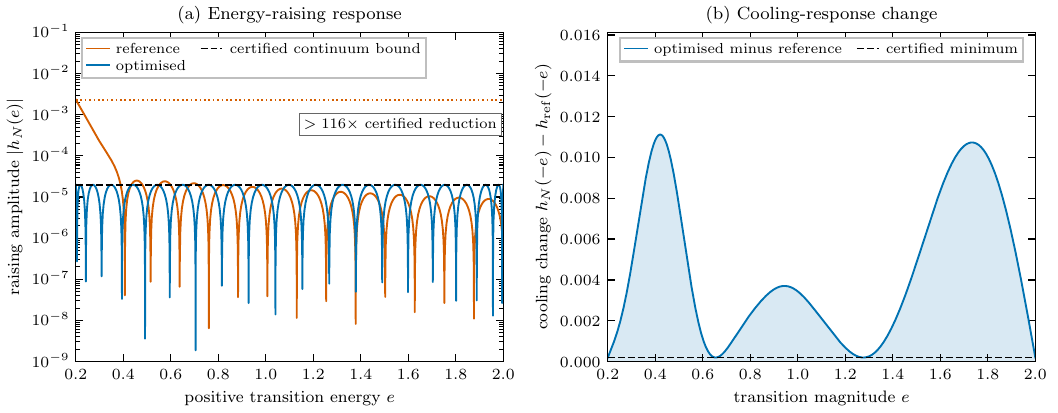}
    \caption{Finite-filter improvement at $\Delta=0.2$, $W=2$, on $N=81$ nodes $t_j=j\pi/5$, $|j|\le40$, with $T=8\pi$ and coefficient one-norm at most $1.470919199$. (a)~Heating response $h_{\rm opt}$ of the optimised filter and its certified continuum bound, which holds at every energy in the band (\cref{app:obs1}), $L_\Delta[h_{\rm opt}] < 1.9881\times10^{-5}$, against the Gevrey response at the gap (dotted), above $2.3073\times10^{-3}$: an improvement greater than $116$. (b)~Difference between the optimised and Gevrey cooling responses; its certified lower bound between samples stays positive, so the optimised filter cools at least as strongly as the Gevrey filter.}
    \label{fig:certificate}
\end{figure*}

\emph{A same-resource certificate.} Before using the support stage, we test the spectral stage on its own against the Gevrey filter with identical resources. On $81$ trapezoidal nodes $t_j = j \pi/5$ for $|j|\le 40$ at $\norm{H}=1$ and $\Delta=0.2$, reweighting the Gevrey coefficients of Ref.~\cite{DingChenLin2024} without increasing their one-norm keeps the cooling at least as strong throughout $[0.2,2]$ and gives $L_\Delta[h_N]<1.9881\times10^{-5}$ on the whole band, compared with the Gevrey value $|h_{\eref,N}(0.2)| > 2.3073 \times 10^{-3}$, an improvement greater than $116$ once the gaps between samples are closed with the second-derivative remainder of \cref{app:obs1}. \Cref{fig:certificate}(a) shows the suppression of the energy-raising transitions, with the certified continuum bound just above the sampled curve, and \cref{fig:certificate}(b) the positive cooling margin.

The optimiser is restricted so that at no point is the new filter able to weaken the cooling response: its cooling floor $P(e)$, the minimum allowed response at $-e$, is the Gevrey filter's own cooling response $h_{\eref,N}(-e)$, raised by $2\times10^{-4}$ in the certificate and the per-model programs (\cref{app:obs1,app:numerics}).
The nodes and the maximum time $T$ are those of the Gevrey filter throughout. The match filters of \cref{sec:numerics}, which minimise the cost at fixed leakage, are constrained only by the Gevrey filter's worst-case leakage, so at an individual transition energy their response can exceed the Gevrey filter's. The Gevrey response is largest at the gap and falls off rapidly above it, while an optimised filter spreads its leakage across the whole band. At transitions well above the gap, where the Gevrey tail has already fallen away, a match filter can therefore heat more strongly. When those transitions set the stationary state, the match filter's stationary infidelity exceeds the Gevrey filter's, as on the eight-site chain of \cref{fig:convergence}(a).

The maximum time cannot be made arbitrarily short. Bernstein's inequality~\cite{Bernstein1912,Boas1954} limits the slope of the response to $T\sup_{\mathbb R}|h_N|$. Across $[-\Delta,\Delta]$ the response must fall by at least the \emph{contrast} $P(\Delta)-L_\Delta[h_N]$, so $T\ge[P(\Delta)-L_\Delta[h_N]]/[2\Delta\sup_{\mathbb R}|h_N|]$, and the maximum time must grow as $1/\Delta$ when the contrast is fixed and the response stays bounded.
Refs.~\cite{DingChenLin2024, LiYangLin2026phase} use this scaling on the maximum time to design their filters.
Ref.~\cite{LiYangLin2026phase} sets the resolution, the $\Delta$ of this bound, to order one, so that the maximum time need only be of order one: a coarser filter spreads the same fall over a wider window, so a shorter maximum time, and hence a cheaper filter, suffices.
Whatever the resolution, the bound is only a lower limit, and it does not show that the leakage gains we observe are the best possible.
Hereafter $\bs{c}^\star$ denotes the selected feasible filter and unqualified jumps and responses use these coefficients.

\emph{Realisable filters.} An ideal filter would have zero response on the whole heating band, so its jumps would never excite the ground state. No nonzero finite filter achieves this: since $h_N(\omega;\bs c)=\sum_j c_j e^{i\omega t_j}$ is analytic in $\omega$ and exponentials with distinct times $t_j$ are linearly independent,
\begin{equation}
    h_N(e;\bs c)=0\ \text{ for all } e\in[\Delta,W]
    \;\Longrightarrow\;
    \bs c=0,
    \qquad\text{so}\qquad
    L_\Delta[h_N]>0\ \text{ for every } \bs c\neq0 .
\end{equation}
The actual transition energies form a finite set, on which further cancellations can occur, but in general a residual must be allowed for through the \emph{raising error}
\begin{equation}
    \eup \coloneqq \max_a \norm{K_{a,N}\Pi_0} = \max_a \norm{K_{a,N}\ket{\psi_0}},
    \label{eq:eupdef}
\end{equation}
the largest norm of a jump applied to the ground state. Its excited-state components are bounded by the heating leakage, while the ground-state component has amplitude $h_N(0)\bra{\psi_0}A_a\ket{\psi_0}$. Hence
\begin{equation}
    \eup \;\le\; L_\Delta[h_N] \;+\; |h_N(0)|\, \max_a \bigl| \bra{\psi_0} A_a \ket{\psi_0} \bigr|.
    \label{eq:eupbound}
\end{equation}

\emph{The zero-frequency term.} The zero-frequency response $h_N(0)=\sum_jc_j$ admits an additional linear constraint. The reference window has $\hat{h}_{\eref}(0)=\tfrac12[\erf(5)-\erf(1)]\approx0.079$. Suppressing this value competes with keeping the heating leakage low, particularly in the Hubbard example below. \textit{Coupling centring} removes the ground-state expectation $\bra{\psi_0}A_a\ket{\psi_0}$ that this response multiplies. For any scalar $\beta$,
\begin{equation}
    \calD[K + \beta \ident] \;=\; \calD[K] \;-\; i\bigl[H_\beta, \,\cdot\,\bigr],
    \qquad
    H_\beta \;\coloneqq\; \tfrac{i}{2}\bigl(\beta^* K - \beta K^\dagger\bigr).
    \label{eq:gauge}
\end{equation}
Thus, for any scalar $\bar A_a$, the shift $A_a \mapsto A_a - \bar A_a \ident$ changes the jump by $\beta = -\bar A_a h_N(0)$ and the generator by a coherent term alone. By \cref{eq:gauge}, the part of $K_{a,N}\ket{\psi_0}$ along $\ket{\psi_0}$, of amplitude $h_N(0)\bra{\psi_0}A_a\ket{\psi_0}$, acts like an extra Hamiltonian term, so it can shift the stationary state by an amount linear in that amplitude. The excited part, of norm at most $L_\Delta[h_N]$, moves population out of the ground state at a rate quadratic in its norm, and it also creates coherences between the ground state and the excited states that are linear in its norm. Both appear in $\sum_a\calD[K_{a,N}](\Pi_0)$, which vanishes exactly when $\Pi_0$ is stationary. Choosing $\bar A_a=\bra{\psi_0}A_a\ket{\psi_0}$ removes the ground-state part exactly but leaves these coherences, and if $\bar A_a$ is only an estimate, the ground-state amplitude $h_N(0)(\bra{\psi_0}A_a\ket{\psi_0}-\bar A_a)$ remains. In \cref{sec:numerics} we measure the effect of centring on a Hubbard chain, where the zero-frequency response is hard to suppress with the filter alone.

\emph{Maintaining the mixing.} The restriction on the displacement $\eta(\bs{c}(\bs{x}), \bs{d}) \le \eta_{\sup}$ of \cref{thm:twostage} keeps the reweighted dynamics convergent.
A generator within distance $\delta$, in the induced trace norm, of one that mixes with pair $(\lambda, \kappa)$ still mixes if $\kappa\delta<\lambda$, with rate at least $\lambda-\kappa\delta$~\cite[Sec.~III.1]{EngelNagel2000}, and its stationary state has moved by at most $\kappa\delta/\lambda$~\cite{SzehrWolf2013}. Two filters $\bs{c}, \bs{d}$ on the same nodes, with generators $\calL(\bs c)$ and $\calL(\bs d)$ of the form \cref{eq:lindblad_K}, have jumps differing by $\sum_j (c_j - d_j)\, e^{iHt_j} A_a e^{-iHt_j}$, of norm at most $\eta(\bs{c},\bs{d})$ since $\norm{A_a} \le 1$. 
The dissipators obey
\begin{equation}
    \indnorm{\calD[K] - \calD[K']} \;\le\; 2\bigl(\norm{K} + \norm{K'}\bigr)\norm{K - K'},
    \label{eq:dissdiff}
\end{equation}
and with $\norm{K_{a,N}(\bs{c})} \le \norm{\bs{c}}_1$, the coherent term common to both, and $m$ jumps,
\begin{equation}
    \indnorm{\calL(\bs{c}) - \calL(\bs{d})} \;\le\; \delta_{c,d} \;\coloneqq\; 2m\bigl(\norm{\bs{c}}_1 + \norm{\bs{d}}_1\bigr)\, \eta(\bs{c},\bs{d}) :
    \label{eq:gendist}
\end{equation}
the generator moves linearly in the displacement the optimiser controls.

\begin{proposition}[Reweighting keeps the mixing]
    \label{prop:transfer}
    If the Gevrey generator $\calL(\bs{d})$ has mixing pair $(\lambda_d, \kappa_d)$ and $\kappa_d\,\delta_{c,d} < \lambda_d$, then $\calL(\bs{c})$ has a unique stationary state $\sigma_c$, mixes with the same prefactor and with rate
    \begin{equation}
        \lambda_c \;\ge\; \lambda_d - \kappa_d\, \delta_{c,d} \;>\; 0,
        \qquad\text{and}\qquad
        \tnorm{\sigma_c - \sigma_d} \;\le\; \frac{\kappa_d}{\lambda_d}\, \delta_{c,d},
        \label{eq:ratetransfer}
    \end{equation}
    where $\sigma_d$ is the stationary state of $\calL(\bs{d})$.
\end{proposition}

The full proof is provided in \cref{app:obs1}. The mixing pair, here and in \cref{thm:stationary}, is a property of an invariant subspace of the working space: both generators compared preserve it, and uniqueness and convergence are asserted on it. Substituting \cref{eq:gendist} into the stability condition $\kappa_d \delta_{c,d} < \lambda_d$ of \cref{prop:transfer} gives
\begin{equation}
    \greq\coloneqq 2\bigl(\norm{\bs{c}}_1 + \norm{\bs{d}}_1\bigr)\, \eta(\bs{c},\bs{d}) \;<\; g_d \;\coloneqq\; \frac{\lambda_d}{\kappa_d\, m}
    \label{eq:margin}
\end{equation}
where $\greq$ bounds each reweighted jump's contribution to the change in the generator. For a fixed safety fraction $0<\gamma<1$, the sufficient \emph{linear} constraint
\begin{equation}
    \eta(\bs c,\bs d)\le\frac{\gamma g_d}{2[\Lambda_{\max}+\norm{\bs{d}}_1]}
    \label{eq:linearcapmain}
\end{equation}
ensures $\greq\le\gamma g_d$; after optimisation we evaluate $\greq$ itself, which can be smaller. The mixing margin $g_d=\lambda_d/(\kappa_dm)$ shrinks as $1/m$ for size-independent $\lambda_d,\kappa_d$.

\subsection{Grading the Radius}

Our second contribution is the observation that each branch depends, up to small corrections, only on the Hamiltonian within its own light cone. A branch evolved for time $t_j$ is generated by the Hamiltonian $H_{B_{R_j}(S_a)} \equiv H_{a,R_j}$, up to an error that the Lieb--Robinson bound makes small once the radius $R_j$ is large enough relative to $v_\LR |t_j|$. For a radius vector $\bs{R} = (R_j)$ we write
\begin{equation}
    K_{a,N}^{(\bs{R})} \coloneqq \sum_{j} c_j^\star\, e^{iH_{a,R_j} t_j} A_a e^{-iH_{a,R_j} t_j}
    \label{eq:mlocaljump}
\end{equation}
for the localised jump, and $\calL_{N,\bs R}$ for the generator with these jumps and unchanged coherent Hamiltonian: ${\mc{L}_H[\rho]=-i[H,\rho]}$.
\begin{proposition}[Branch-adaptive locality, informal]
    \label{prop:locality}
    Under the standard uniform Lieb--Robinson and growth assumptions, a constant $C_4$, fixed by the interaction bounds, graph growth, coupling-support radius $r_S$ and locality constants, bounds the \emph{spatial error} for any radii:
    \begin{equation}
        \bigl\lVert K_{a,N} - K^{(\bs{R})}_{a,N} \bigr\rVert \;\le\; \esp(\bs{R}; \bs{t}, \bs{c}^\star) \;\coloneqq\; \sum_{j} \bigl| c_j^\star \bigr| \min\Bigl\{ 2,\; C_4 |t_j|\, e^{\mu v_\LR |t_j|}\, e^{-\frac{\mu}{2}(R_j - 1)} \Bigr\},
        \label{eq:espmain}
    \end{equation}
    and the localised and global generators differ, in the induced trace norm, by at most $\delta_{N,\bs{R}} \le 4 m \norm{\bs{c}^\star}_1\, \esp$. In particular, giving every nonzero-time branch the radius
    \begin{equation}
        R_j(\ell) \;\coloneqq\; 1 + \lceil 2 v_\LR |t_j| + \ell \rceil,
        \label{eq:adaptive}
    \end{equation}
    a fixed padding $\ell$ beyond its own light cone, makes the error a single exponential in the padding,
    \begin{equation}
        \esp \;\le\; C_4\, e^{-\mu\ell/2} \sum_j |c_j^\star|\,|t_j|.
        \label{eq:adaptivemain}
    \end{equation}
\end{proposition}
\Cref{app:obs2} provides the proof of the bound. The \emph{light-cone schedule} $\bs R(\ell)$ has maximum radius $R(T)\coloneqq 1+\lceil2v_\LR T+\ell\rceil$. At fixed padding, its radii are determined before coefficient optimisation, so removing long branches also removes expensive patches.

\begin{figure}[b]
    \centering
    \includegraphics[width=0.35\linewidth]{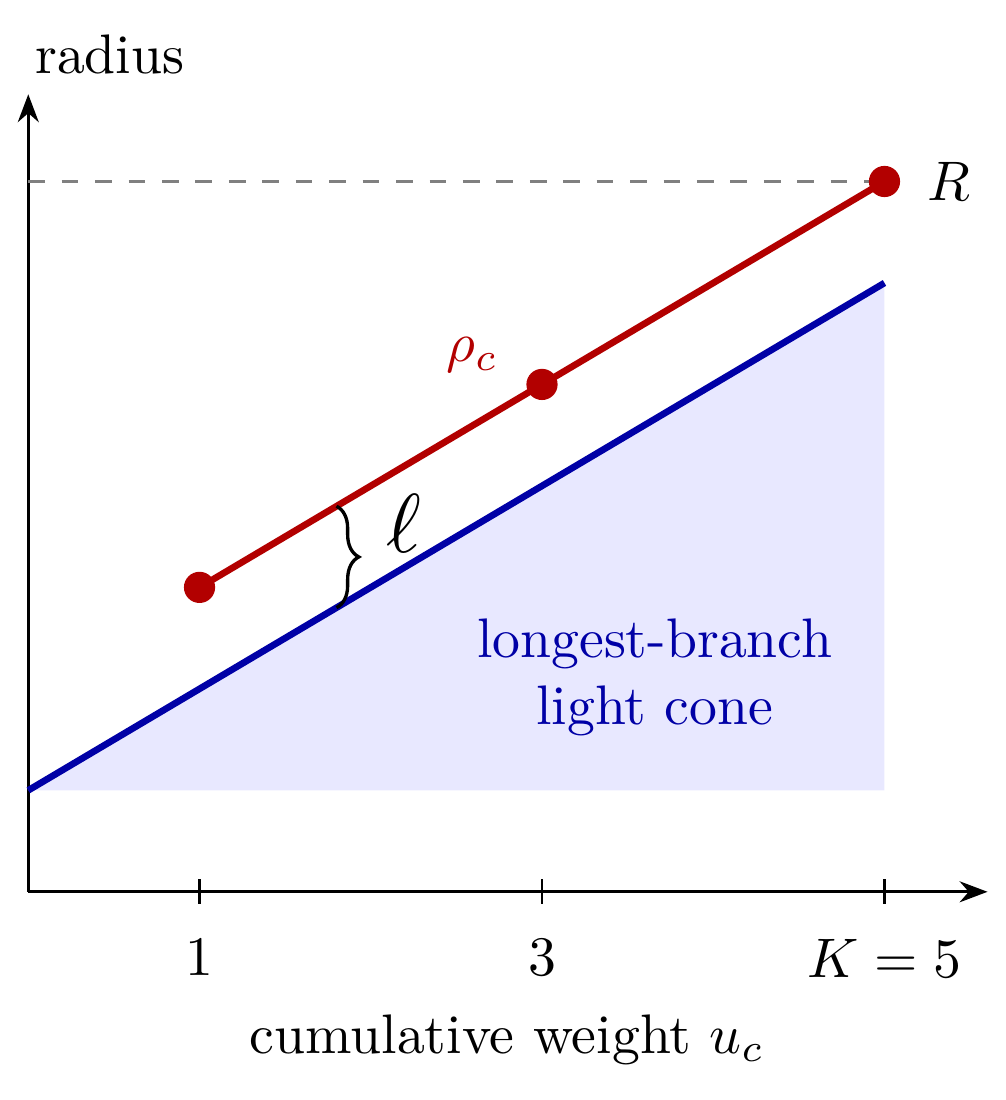}
    \caption{Radius grading by cumulative evolution time. Ordered weights $(w_1,w_2,w_3)=(1,2,2)$ sum to $K=5$ and give cumulative weights $(u_1,u_2,u_3)=(1,3,5)$, where $u_c=\sum_{c'\le c}w_{c'}$. Red points assign each shared evolution piece the radius $\rho_c=1+\lceil2v_\LR u_c\tau+\ell\rceil$, with padding $\ell$ above the blue longest-branch light cone, up to the longest-branch radius $R(T)$ at $u_c=K$.}
    \label{fig:construction}
\end{figure}

\emph{Sharing evolutions between branches.} We implement the jump operator with a \emph{compact additive-digit ladder}, a block-encoding circuit that builds every branch of the jump from a single \textit{shared} set of evolutions. Its time register holds a sign bit and $b$ magnitude bits, or digits; digit $c$ controls an evolution of duration $w_c\tau$, where the weights $w_1,\ldots,w_b$ are positive integers summing to $K$. The sign-and-magnitude time register is that of Refs.~\cite{Chen2023thermal,Chen2023gibbs}, which apply the controlled evolution at integer multiples of a base step; the digit weights, their radii and the sharing of evolutions between branches are introduced in this work.

A weight list is \emph{complete} if every magnitude $j=0,\ldots,K$ is the sum of some subset of the weights. For each magnitude we fix one such subset, recorded as a bit string $\bs\chi(j)\in\{0,1\}^b$ with $\sum_cw_c\chi_c(j)=j$. The branch labelled $sj$, with sign $s=\pm1$ and magnitude $j$, is conjugated by the digits it selects, applied in order: $V_{s,j}=(U_b^{(s)})^{\chi_b(j)}\cdots(U_1^{(s)})^{\chi_1(j)}$, with $U_c^{(s)}=e^{isw_c\tau H_{a,\rho_c}}$ the unitary evolution of the Hamiltonian over a patch of radius $\rho_c$. Conjugating the coupling therefore schedules $2b$ controlled evolutions, of total duration $2K\tau$, shared by all branches; generating the nonzero branches independently would take total duration $2\tau K(K+1)$.

Digit $c$ acts after digits $1,\ldots,c-1$, so by its end every branch that uses it has evolved for at most $u_c\tau$, where $u_c=\sum_{c'\le c}w_{c'}$ is the cumulative weight. Giving digit $c$ the padded light-cone radius of that elapsed time,
\begin{equation}
    \rho_c\coloneqq1+\lceil2v_\LR u_c\tau+\ell\rceil,\qquad 1\le c\le b,
    \label{eq:compactradii}
\end{equation}
therefore preserves \cref{eq:adaptivemain}. \Cref{fig:construction} shows the radii growing with cumulative weight at a fixed padding $\ell$. In the numerics we fix the minimum width $b=\lceil\log_2(K+1)\rceil$ and search over complete weight lists and their order for the cheapest ladder. For a shortened support, whose largest retained magnitude is $K_{\rm act}$, the ladder uses $K_{\rm act}$ in place of $K$, while the label overhead still counts all $N$ candidate labels.

More generally, a branch can be compiled as any ordered product of $p_j$ evolution pieces, $V_j=U_{j,p_j}\cdots U_{j,1}$ with $U_{j,p}=e^{i\operatorname{sgn}(t_j)H_{a,\rho_{j,p}}\theta_{j,p}}$, where the positive durations $\theta_{j,p}$ sum to $|t_j|$ and piece $p$ ends at elapsed time $u_{j,p}=\sum_{p'\le p}\theta_{j,p'}$. Its light-cone error is then at most $\varepsilon_j$, and the spatial error at most $\esp$, where
\begin{equation}
    \varepsilon_j\coloneqq\min\!\left\{2,\sum_pC_4\theta_{j,p}e^{\mu v_\LR u_{j,p}-\mu(\rho_{j,p}-1)/2}\right\},\qquad \esp=\sum_j|c_j^\star|\varepsilon_j,
    \label{eq:wordmain}
\end{equation}
and $\varepsilon_0=0$ for the zero-time branch. \Cref{eq:wordmain} extends \cref{eq:espmain} to the compiled circuit. It is an upper bound, so it cannot say which of two compilations has the smaller actual error.

Coefficient preparation loads each label $j$ with amplitude $\sqrt{|c_j^\star|/\norm{\bs{c}^\star}_1}$, writing its sign to the sign bit ($0$ for $j\ge0$, $1$ for $j<0$), which sets the direction $s$ of every evolution, and $\bs\chi(|j|)$ to the digits. Applying a diagonal phase $e^{i\arg c_j^\star}$ and the controlled, conjugated coupling, and then undoing the preparation, block-encodes the compiled jump divided by $\norm{\bs{c}^\star}_1$. A coupling that is not unitary needs its own block encoding, and coefficient preparation and phase synthesis are further circuit components.

We estimate the cost of each block encoding of a filtered jump in Toffoli gates. Let $\mathcal P=\{(\theta_p,\rho_p)\}$ list the evolution pieces of one side of the conjugation, piece $p$ running for time $\theta_p$ on a patch of radius $\rho_p$; the other side repeats them in reverse, which gives the factor $2$ in
\begin{equation}
    \begin{split}
        Q_p&\coloneqq\lceil\lambda_H(\rho_p)\theta_p\rceil+\lceil\log_2(N/\eHS)\rceil,\\
        C_{\rm BE}(\mathcal P)&\coloneqq A_{\rm norm}\left[2\sum_{p\in\mathcal P}Q_pC_H(\rho_p)+(N-2)+2\lceil\log_2N\rceil\right].
    \end{split}
    \label{eq:costmain}
\end{equation}
Piece $p$ makes $Q_p$ queries to a block encoding of the patch Hamiltonian $H_{a,\rho_p}$, whose nonidentity Pauli coefficients have one-norm $\lambda_H(\rho_p)$. Each query costs $C_H(\rho)=L_t(\rho)-2+2C_{\rm prep}(\rho)+2\lceil\log_2L_t(\rho)\rceil$ Toffolis, where $L_t(\rho)$ is the number of Pauli terms in the patch Hamiltonian and $C_{\rm prep}(\rho)$ the cost of preparing its coefficients. The surcharge $\lceil\log_2(N/\eHS)\rceil$ adds a fixed number of queries to every piece and is the only place the cost-model parameter $\eHS=10^{-3}$, a nominal simulation precision, enters; implementation accuracy is set separately, through $\zeta_{\rm imp}$ below. The terms $(N-2)+2\lceil\log_2N\rceil$ charge label selection and reflection over all $N$ candidate labels, even when the support is shortened. The multiplier $A_{\rm norm}\coloneqq\max\{1,\lceil\norm{\bs{c}^\star}_1\rceil\}$ is an accounting convention. For the ladder, $\mathcal P=\{(w_c\tau,\rho_c)\}_{c=1}^{b}$, and grading changes the patch cost of each piece while keeping the $2b$ controlled evolutions.

Physical energy units must scale together: for a shift $E_{\rm mid}$ and scale $\alpha>0$, with $\bar H=(H_{\rm phys}-E_{\rm mid}\ident)/\alpha$,
\begin{equation}
    \bar\Delta=\Delta_{\rm phys}/\alpha,\quad \bar\lambda_H=\lambda_H^{\rm phys}/\alpha,\quad \bar v_\LR=v_\LR^{\rm phys}/\alpha,\quad t_{\rm phys}=\bar t/\alpha.
    \label{eq:energyscalingmain}
\end{equation}
The scalar shift cancels in conjugation, and $\bar\lambda_H\bar t=\lambda_H^{\rm phys}t_{\rm phys}$. In the two-dimensional cost comparisons the gap, bandwidth, patch norms $\lambda_H(\rho)$ and locality constants are chosen independently of one another, to show how the compiled cost responds to each. Their query-cost functions are specified with the results. \Cref{app:frontier} gives the rest of the accounting: the Toffoli cost and work qubits of each patch-Hamiltonian query, the cost of separately generated branches and of the ladder, the $T$-gate, rotation and ancilla-qubit counts, including preparation of the coefficients by coherent alias sampling, the scaling with system size and gap, the settings behind \cref{tab:resources}, and the costs left out, such as the outer loop and rotation synthesis.

\subsection{Convergence Survives Truncation}

Our third contribution is the observation that convergence survives truncation. The light-cone schedule grows with the branch time, and for the long branches of a small-gap filter it reaches the whole system. For a maximum radius $R^*$, let $\bs{R}^*$ be the light-cone schedule truncated at $R^*$,
\begin{equation}
    (\bs{R}^*)_j \;\coloneqq\; \min\{R_j(\ell),\, R^*\},
    \label{eq:cap}
\end{equation}
and $\calL_{N,\bs{R}^*}$ the resulting generator. The distance $\delta_{N,\bs{R}^*}=\indnorm{\calL_{N,\bs{R}^*}-\calL_N}$ obeys $\delta_{N,\bs{R}^*}\le4m\norm{\bs{c}^\star}_1\esp(\bs{R}^*)$, because both ideal jumps have norm at most $\norm{\bs{c}^\star}_1$. Applying the perturbation argument again gives the following stationary-state guarantee.

\begin{theorem}[Stationary state of the truncated dynamics; conditional on a mixing pair]
    \label{thm:stationary}
    If the reweighted global generator $\calL_N$ has mixing pair $(\lambda_N, \kappa_N)$ and $\kappa_N\,\delta_{N,\bs{R}^*} < \lambda_N$, then the truncated generator $\calL_{N,\bs{R}^*}$ mixes, has a unique stationary state $\sigma_{N,\bs{R}^*}$, and that state lies close to the ground state:
    \begin{equation}
        \tnorm{\sigma_{N,\bs{R}^*} - \Pi_0} \;\le\; \frac{\kappa_N}{\lambda_N} \Bigl[ m\bigl(\eup^2 + \norm{\bs{c}^\star}_1\, \eup\bigr) + \delta_{N,\bs{R}^*} \Bigr].
        \label{eq:mainbound}
    \end{equation}
\end{theorem}
\Cref{app:obs3} gives the proof.
For the implementation, write $K_{a,N}^{(\bs R)}=\sum_jc_j^\star V_jA_aV_j^\dagger$ for the ideal compiled jump, with the compiled conjugators $V_j$ above in place of $e^{iH_{a,R_j}t_j}$ in \cref{eq:mlocaljump}, and $\widetilde K_{a,N}^{(\bs R)}=\sum_j\widetilde c_j\widetilde V_jA_a\widetilde V_j^\dagger$ for the jump the circuit applies, with rounded coefficients $\widetilde c_j$ and approximate unitary conjugators $\widetilde V_j$. We assume the implemented jumps preserve the same working space, and let $\zeta_{\rm imp}\ge\max_a\norm{\widetilde K_{a,N}^{(\bs R)}-K_{a,N}^{(\bs R)}}$ bound the implementation error. Since $\norm{A_a}\le1$, a sufficient allowance is
\begin{equation}
    \zeta_{\rm imp}\ge\sum_j|\widetilde c_j-c_j^\star|+2\sum_j|c_j^\star|\norm{\widetilde V_j-V_j}.
    \label{eq:implementationmain}
\end{equation}
The allowance $\zeta_{\rm imp}$ also absorbs any further preparation or block-encoding error, measured after undoing the normalisation. The implemented jump norm can reach $\norm{\bs{c}^\star}_1+\zeta_{\rm imp}$, giving extra generator distance $m(4\norm{\bs{c}^\star}_1+2\zeta_{\rm imp})\zeta_{\rm imp}$. If $4\norm{\bs{c}^\star}_1(\esp+\zeta_{\rm imp})+2\zeta_{\rm imp}^2<g_d-\greq$, the generator with the implemented jumps mixes, and its stationary state $\widetilde\sigma$ obeys
\begin{equation}
    \tnorm{\widetilde\sigma-\Pi_0}\le\frac{\eup^2+\norm{\bs{c}^\star}_1\eup+4\norm{\bs{c}^\star}_1(\esp+\zeta_{\rm imp})+2\zeta_{\rm imp}^2}{g_d-\greq}.
    \label{eq:statescoremain}
\end{equation}
Here $\greq=2[\norm{\bs{c}^\star}_1+\norm{\bs{d}}_1]\eta(\bs c^\star,\bs d)$ and $g_d$ is the reference mixing margin.

The heating, truncation and implementation errors are all multiplied by the same factor, $\kappa_N/\lambda_N$. For a fixed mixing pair, each jump's share of the error must therefore shrink as $1/m$ to certify a fixed global error. Because the spatial tails decay exponentially, the padding need only grow as $\log m$; the implementation precision and the raising error must also improve. When the truncated maximum radius $R^*$ falls below the padded light-cone radius of the longest branch $R(T)$, the branches fall into two groups. A branch whose own padded light cone fits inside the cap, $R_j(\ell)\le R^*$, is generated on its full patch and keeps the exponentially small error of \cref{eq:adaptivemain}. A longer branch is cut off inside its light cone, where the Lieb--Robinson bound gives no useful control, so the only bound on its error is the trivial one: the exact and truncated couplings each have norm at most one, so they differ by at most two. Weighting each branch by its coefficient gives
\begin{equation}
    \esp(\bs{R}^*) \;\le\; C_4\, e^{-\mu\ell/2} \sum_{j : R_j(\ell) \le R^*} |c_j^\star|\,|t_j| \;+\; 2 \sum_{j : R_j(\ell) > R^*} |c_j^\star|.
    \label{eq:capbound}
\end{equation}
The first sum shrinks exponentially as the padding grows. The second does not shrink at all: it is twice the coefficient weight carried by the truncated branches, so unless those branches carry little weight it dominates, and the bound cannot certify a cap below the light cone. The chains of \cref{sec:numerics} are in exactly this situation, yet the truncated Ising and Hubbard dynamics remain accurate, so the bound is far from tight there.

The bound is also global. It controls the whole state in trace norm, and through the factor $m$ in $\delta_{N,\bs R^*}$ it grows with the number of jumps, and so with the system size. A guarantee that does not weaken as the system grows would have to control local observables instead, which requires a local stability result for dissipative dynamics, of the kind proved in Ref.~\cite{Cubitt2015} under rapid mixing, and goes beyond the global bound used here.

\emph{Stopping rule.} The gap also yields a stopping rule. On the working space,
\begin{equation}
    \Pi_{\mathrm S}(H - E_0 \ident)\Pi_{\mathrm S} \;\succeq\; \Delta\,(\Pi_{\mathrm S} - \Pi_0)
    \qquad\text{and hence}\qquad
    1 - \Tr[\Pi_0 \rho] \;\le\; \frac{\Tr[H\rho] - E_0}{\Delta}
    \label{eq:energycertmain}
\end{equation}
for every state $\rho=\Pi_{\mathrm S}\rho \Pi_{\mathrm S}$. At test $r$, use independent preparations and an energy upper confidence bound $\overline E_r$, a ground-energy lower bound $\underline E_0$, and a positive gap lower bound $\underline\Delta$, all on this space. Accept when $(\overline E_r-\underline E_0)/\underline\Delta\le\varepsilon$; this certifies population at least $1-\varepsilon$ and trace-norm error at most $2\sqrt\varepsilon$, that is, trace distance at most $\sqrt\varepsilon$. With valid bounds every acceptance is correct, although a loose $\underline E_0$ or $\underline\Delta$ can reject a state that meets the target. Give test $r$ a failure probability $\delta_r$, covering all of its estimates and any uncertainty in the reference bounds, with $\sum_r\delta_r\le\delta_{\rm fail}$; by the union bound, the probability that any test wrongly accepts is then at most $\delta_{\rm fail}$. Measurement is the other cost of the test, and these allowances set it. If each energy estimate averages single-shot outcomes of bounded range, Hoeffding's inequality~\cite{Hoeffding1963} gives an upper confidence bound within $\varepsilon_E$ of the true energy, with failure probability $\delta_r$, from $O(\varepsilon_E^{-2}\log(1/\delta_r))$ preparations, with a constant set by that range. To accept, the bound must be tighter than the margin by which the true energy excess falls below $\varepsilon\underline\Delta$. When the excess sits a fixed fraction below that threshold, a test therefore costs of order $(\varepsilon\underline\Delta)^{-2}\log(1/\delta_r)$ preparations, each of which runs the dynamics afresh, and more as the state approaches the threshold; spreading $\delta_{\rm fail}$ over more tests raises this only logarithmically.

The test certifies whatever state the computation actually produces, so it remains valid whatever errors the implementation makes. It does not, however, ensure that an acceptable state is ever produced. That needs two further guarantees. The dynamics must converge, which under a mixing pair takes time of order $\lambda_N^{-1}\log(\kappa_N/\varepsilon)$ (\cref{eq:tstop}). The outer loop, the algorithm that simulates the Lindblad dynamics by calling the block-encoded jumps, must also reproduce that dynamics closely enough that its own error, from splitting the channel across jumps and from finite time steps, fits within the remaining allowance; \cref{tab:outerloop} shows the size of the splitting error on the six-site chain. Neither is costed in this paper, which prices the jump primitive alone (\cref{app:frontier}).

\begin{figure*}[b]
    \includegraphics[width=\textwidth]{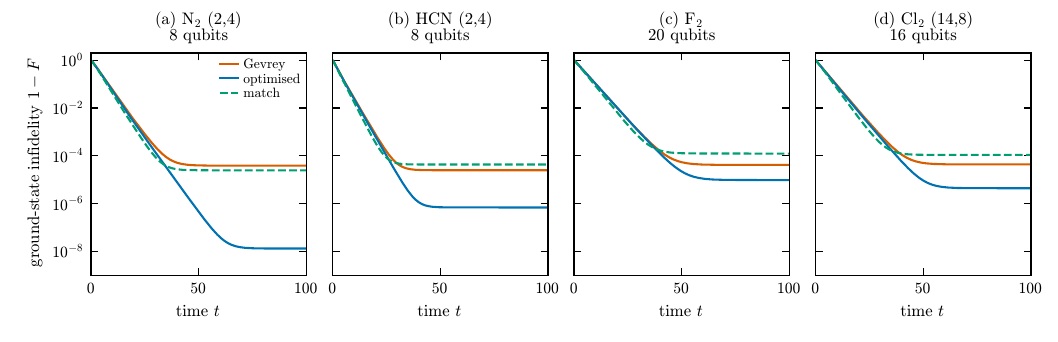}
    \caption{Ground-state infidelity $1-F(t)$, $F(t)=\langle\psi_0|\rho(t)|\psi_0\rangle$, under Lindblad evolution from the highest-energy state of the fixed-electron-number, $S_z=0$ sector, for the Gevrey filter, the optimised filter at equal cost and the match filter at matched design-grid leakage (whole-band for F$_2$), all with full-Hamiltonian branches; a molecular Hamiltonian has no locality to grade, so there is no patch panel. (a)~N$_2$ and (b)~HCN in $(2,4)$ active spaces. (c)~F$_2$ in its full STO-3G space, sector dimension $100$. (d)~Cl$_2$ in a $(14,8)$ active space, sector dimension $64$. The F$_2$ match curve uses the re-solved coefficients of \cref{app:numerics}. All panels show $0\le t\le100$; the F$_2$ and Cl$_2$ runs continue to $t=160$, where $\tnorm{\calL[\rho]}$ is below $10^{-10}$ and $F$ changes by less than $10^{-10}$ over the final five time units.}
    \label{fig:molconvergence}
\end{figure*}

\section{Implementation and numerics}
\label{sec:numerics}
\label{sec:impact}

The linear program of \cref{thm:twostage} allows two optimisation regimes: we may minimise the leakage at fixed cost, or minimise the cost at fixed leakage. The \emph{optimised filter} keeps every branch, so it costs the same as the reference, and minimises the heating leakage. The \emph{match filter} keeps the reference's leakage and minimises the cost: it is the shortest symmetric set of branches whose heating leakage on the design grid is no larger than the reference's. Both new filters cool at least as strongly as the Gevrey filter at every energy where the cooling constraint is imposed. For each model we fix the time grid, the cooling floor and the budgets on the coefficient norm and displacement, optimise the coefficients, and then price every feasible support and ladder weight list with the cost model of \cref{eq:costmain}. The cost columns of \cref{tab:molecules,tab:lattice} give the percentage change in this ladder cost relative to the Gevrey filter compiled the same way, with gaps, coefficient norms and times in the units of the normalised Hamiltonian.

\subsection{Molecules: Testing the Optimised Filter}
We use the fourteen molecules to which Ref.~\cite{LiZhanLin2025ab} applied dissipative preparation, counting the chain and square geometries of H$_4$ separately, for the fifteen rows of \cref{tab:molecules}. All molecules presented in this work use the STO-3G basis; a label such as LiH $(2,4)$ gives the numbers of active electrons and active spatial orbitals, and molecules without a label use their full orbital space. Molecular orbitals are delocalised, so the interaction graph is complete and every patch of positive radius contains the whole molecule. Each branch therefore simulates the full Hamiltonian, and only the first observation, the optimisation of coefficients and support, applies: the molecules test the linear optimisation of the filter on its own. We do not rule out Lieb--Robinson radius grading on larger molecular targets with geometrically bounded interactions. Gaps are computed in the sector with fixed electron number and $S_z=0$, and costs use the Pauli-term count and coefficient one-norm of the normalised Jordan--Wigner Hamiltonian. For seven of the molecules, H$_2$, the H$_4$ chain, LiH, N$_2$, HCN, F$_2$ and Cl$_2$, we also simulate the full density-matrix dynamics, with norm-one spin-resolved hopping couplings $A_{ii'\sigma}\propto a^\dagger_{i\sigma}a_{i'\sigma}+\mathrm{h.c.}$, $i<i'$, which preserve this sector while connecting different total-spin sectors.

\begin{table*}[t]
    \caption{Molecular filter design in STO-3G. $n$ is the number of qubits; a pair such as $(2,4)$ after a molecule gives the numbers of active electrons and spatial orbitals. The physical gap is in hartree in the fixed-electron-number, $S_z=0$ sector, and $\Delta$ is that gap divided by the sector half-bandwidth. $N$ and $T$ are the reference node count and realised maximum time. The optimised filter keeps the full support, and the match filter is the shortest symmetric support meeting the reference leakage on the design grid; the F$_2$ and H$_{10}$ match vectors are re-solved to meet it on the whole band (\cref{app:numerics}). Leakage gain is $L_\Delta[h_{\eref,N}]/L_\Delta[h_{\rm opt}]$, with $h_{\rm opt}$ the optimised filter's response, evaluated on a fine check grid of $20{,}000$ energies. Active is $N_{\rm act}/N$, the fraction of branches retained. The cost change is $100(C_{\rm match}/C_{\eref,N}-1)\%$, where $C_{\rm match}$ and $C_{\eref,N}$ are the ladder costs $C_{\rm BE}$ of the match and Gevrey filters; a negative value is a saving. H$_4$ chain and square have $2.0$~\AA{} spacing, and the H$_6$ and H$_{10}$ chains have $0.7$~\AA{} spacing.}
    \label{tab:molecules}
    \begin{ruledtabular}
        \small
        \begin{tabular}{lcccccccc}
             & & & & & & & \multicolumn{2}{c}{match} \\
            \cline{8-9}
            molecule & $n$ & gap (Ha) & $\Delta$ & $N$ & $T$ & leakage gain & active & cost change \\
            \colrule
            \MolTableRows
        \end{tabular}
    \end{ruledtabular}
\end{table*}

At the reference leakage, the match filter drops the longest branches and lowers the ladder cost for every molecule (\cref{tab:molecules}). It does not reduce the overhead of the label register, whose size and preparation are set by the number of candidate labels $N$. Each shortened support is re-optimised against the reference leakage. For most molecules this constraint is imposed on the design grid, and a check on a much finer grid shows how far the leakage exceeds the cap between grid points. For F$_2$ and H$_{10}$ we instead re-solve by constraint generation~\cite{Kelley1960}, adding violated energies until the error bounds between samples guarantee the reference leakage, and a cooling response at least $2\times10^{-4}$ above the Gevrey floor, across the whole band.

\Cref{fig:molconvergence} shows the ground-state infidelity against time for four of the molecules, from density-matrix simulations of the Lindblad dynamics. The optimised filters reach a lower infidelity than the Gevrey filter on all four. The match filters do so only for N$_2$; on HCN, F$_2$ and Cl$_2$ their infidelity is about $1.7$ to $3$ times higher. Equal worst-case leakage still allows different responses at the particular transition energies that set the stationary state. Matching the worst-case leakage therefore lowers the cost without necessarily matching the accuracy, which depends on the molecule.

\subsection{Chains: Models and Filter Design}
The chain numerics test filter design, locality and relaxation on three models, each with open boundaries. The first is the transverse-field Ising chain of Ding, Chen and Lin~\cite{DingChenLin2024}, $H=-\sum_iZ_iZ_{i+1}-h\sum_iX_i$ with $h=1.2$, at four, six and eight sites. The second is a six-site Heisenberg chain in a longitudinal field, $H=\sum_i(X_iX_{i+1}+Y_iY_{i+1}+Z_iZ_{i+1})+h\sum_iZ_i$, at a larger field $h=2$ and a smaller field $h=0.5$. The third is the one-dimensional Hubbard model of Ref.~\cite{DingChenLin2024}, $H=-t_{\rm hop}\sum_{i,\sigma}(a^\dagger_{i\sigma}a_{i+1,\sigma}+\mathrm{h.c.})+U\sum_in_{i\uparrow}n_{i\downarrow}$, on four sites with $t_{\rm hop}=1$ and $U=4$, mapped to eight qubits by the Jordan--Wigner transformation and restricted to the $36$-dimensional sector with $N_\uparrow=N_\downarrow=2$. Each Hamiltonian is shifted and scaled so that its spectrum lies in $[-1,1]$, and the filters are designed for its normalised gap $\Delta$, listed in \cref{tab:lattice}. Every Ising or Heisenberg site carries a jump coupling, $Z_a$ or $X_a$ respectively, and every Hubbard bond carries one normalised, spin-summed hopping coupling. The Gevrey filter is the erf window of Ref.~\cite{DingChenLin2024} at the model's own $\Delta$, with the node count and maximum time of that prescription.

Each simulation evolves the density matrix under the full Lindblad master equation, starting from the highest-energy eigenstate, which has no overlap with the ground state, and records the ground-state fidelity and the energy; \cref{app:numerics} gives the integration details. For the four- and six-site chains and for the Hubbard chain, we also compute the stationary state directly, from the null space of the Lindbladian. For Hubbard this null space is seven-dimensional, because the hopping jumps conserve total spin and total pseudospin and annihilate the $S_z=0$ member of the fully polarised quintet, so we take the stationary state reached from our initial state. The mixing diagnostics of \cref{app:numerics} are computed in the ten-dimensional spin-singlet, pseudospin-singlet subspace that contains the ground and initial states, where the stationary state is unique and the gap exceeds the design $\Delta$.

\begin{table*}[b]
    \caption{Chain filter design and dynamical outcomes. $L$ is the site count, $h$ the field strength, opt the full-support optimised filter and match the shortest symmetric support meeting the reference leakage on the design grid. $\Delta$ is the spectral gap after scaling the Hamiltonian to $[-1,1]$; $N,T$ are the reference node count and realised maximum time. $L_\Delta$ is the maximum heating amplitude on the fine check grid, in units of $10^{-3}$. Active and the cost change are as in \cref{tab:molecules}, so negative entries save work.
    $F=\langle\psi_0|\rho|\psi_0\rangle$ is the stationary ground-state population, except for the eight-site row, which gives $F(40)$. Patch columns use the match filter with the branch-dependent generator $\mathcal L^{[R^*]}$ of \cref{eq:cappedpatchmain}; $R^*$ is the radius cap, and $R^*=0$ keeps only the coupling support, a single site for spins and a bond for Hubbard. Each trajectory starts in the highest-energy state of its invariant sector; a dash denotes an uncomputed case.}
    \label{tab:lattice}
    \label{tab:lattice-dyn}
    \begin{ruledtabular}
        \footnotesize\setlength{\tabcolsep}{2.4pt}
        \begin{tabular}{lcccccccccccc}
             & & & & & & & \multicolumn{2}{c}{$F$, full $H$} & \multicolumn{4}{c}{$F$, match, radius cap $R^*$}\\
            \cline{8-9}\cline{10-13}
            model & $\Delta$ & $N$ & $T$ & $L_\Delta$: Gevrey $\to$ opt & active & cost & Gevrey & match & $0$ & $1$ & $2$ & $3$\\
            \colrule
            TFIM, $L=4$ & $0.189$ & 85 & $26$ & $2.3 \to 0.046$ & 37/85 & $-27\%$ & $0.99998$ & $0.99998$ & --- & $0.9727$ & --- & --- \\
            TFIM, $L=6$ & $0.097$ & 165 & $52$ & $2.3 \to 0.089$ & 71/165 & $-30\%$ & $0.99999$ & $0.99998$ & --- & $0.9334$ & $0.9894$ & --- \\
            TFIM, $L=8$ & $0.062$ & 257 & $80$ & $2.3 \to 0.213$ & 113/257 & $-39\%$ & $0.99999$ & $0.99998$ & --- & $0.8936$ & $0.9750$ & $0.9957$ \\
            Heisenberg, $L=6$, $h=2$ & $0.106$ & 151 & $47$ & $2.3 \to 0.081$ & 65/151 & $-32\%$ & $0.99991$ & $0.99958$ & --- & $0.2804$ & $0.7329$ & --- \\
            Heisenberg, $L=6$, $h=0.5$ & $0.108$ & 149 & $46$ & $2.3 \to 0.080$ & 65/149 & $-32\%$ & $0.99989$ & $0.99972$ & --- & $0.4261$ & $0.2905$ & --- \\
            Hubbard, $L=4$ (8 qubits) & $0.091$ & 177 & $55$ & $2.3 \to 0.096$ & 77/177 & $-32\%$ & $0.99829$ & $0.99417$ & $0.6716$ & $0.9680$ & --- & ---\\
        \end{tabular}
    \end{ruledtabular}
\end{table*}

On the chains, the optimised filter lowers the heating leakage by a factor of \LeakGainMin{} to \LeakGainMax{} (\cref{tab:lattice}). Its stationary infidelity on the spin chains is \OptInfidMin{} to \OptInfidMax{}, against \RefInfidMin{} to \RefInfidMax{} for the Gevrey filter; for the eight-site chain we quote the value at the final time. The Hubbard chain is the exception: there the optimised filter leaves \HubInfidOpt{} against the Gevrey filter's \HubInfidRef{}, because of the zero-frequency term discussed below. Held to the Gevrey filter's leakage, the match filter keeps only the shortest \MatchKeptMin{} to \MatchKeptMax{} of the branches and lowers the per-jump cost by \CostGainMin{} to \CostGainMax{}. These costs use full-Hamiltonian branches at a uniform radius, and each ladder digit pays the surcharge $\lceil\log_2(N/\eHS)\rceil=17$ or $18$ once; the square-lattice tests below grade the radii of the digits. The curves in \cref{fig:convergence}(a) relax at similar rates at first and then settle onto plateaux that depend on the filter.

\begin{table}[t]
    \caption{Effect of zero-frequency response on the four-site Hubbard chain. $h_N(0)=\sum_jc_j$ is the filter response at zero energy; $L_\Delta$ is its maximum heating amplitude. $1-F_\infty$ is stationary ground-state infidelity. The three columns use couplings $A_a$, $A_a-\tfrac12\ident$, and $A_a-\langle A_a\rangle\ident$, where $\langle A_a\rangle=\langle\psi_0|A_a|\psi_0\rangle$. The coherent Hamiltonian is held fixed. A dash marks an uncomputed or unreported entry.}
    \label{tab:hubdc}
    \begin{ruledtabular}
        \begin{tabular}{lccccc}
             & & & \multicolumn{3}{c}{$1 - F_\infty$} \\
            \cline{4-6}
            filter & $h_N(0)$ & $L_\Delta$ & $A_a$ & $A_a - \tfrac12$ & $A_a - \langle A_a \rangle$ \\
            \colrule
            \HubDCRows
        \end{tabular}
    \end{ruledtabular}
\end{table}

\emph{Centring the coupling.} The Hubbard chain is where the zero-frequency term of \cref{eq:eupbound} matters. A jump applied to the ground state returns, besides the excitations that the leakage controls, a component along the ground state itself, of amplitude $h_N(0)\bra{\psi_0}A_a\ket{\psi_0}$, which can shift the stationary state linearly. The Pauli couplings of the spin chains have zero ground-state expectation by symmetry, but the Hubbard bond coupling has an expectation of about a half. \Cref{tab:hubdc} lists each filter's zero-frequency response $h_N(0)$ beside its stationary infidelity, and the two go together: every filter leaves an infidelity of a few parts in a thousand, and the match filter, with the largest $h_N(0)$, leaves the most. Suppressing $h_N(0)$ through the filter alone is costly. Imposing $|h_N(0)|\le0.04$ in the spectral program, even with the displacement budget lifted so that the constraint can be met at all, raises the heating leakage to $\HubDCcapEps$, and $|h_N(0)|\le0.01$ is infeasible on this grid. Centring the coupling works much better. Replacing $A_a$ by $A_a-\bra{\psi_0}A_a\ket{\psi_0}$, the shift of \cref{eq:gauge} applied to the coupling without compensating the generator, brings the Gevrey filter's stationary infidelity down to $\HubResidualShift$; replacing it by $A_a-\tfrac12$, which needs no knowledge of the ground state, lowers every filter's infidelity by a factor of four. The hopping couplings of H$_2$, the H$_4$ chain, LiH, N$_2$ and HCN have small ground-state expectations, and their density-matrix simulations reach stationary infidelities below $10^{-4}$ without any shift.

\subsection{Chains: Truncating to Patches}
In the locality experiment each branch is generated on a patch whose radius grows with its evolution time, up to a cap $R^*$. For a chain, the irreducible-path form of the Lieb--Robinson bound~\cite{ChenLucasYin2023} gives the velocity $v_\LR=2\mathrm eJ/\alpha$ at $\mu=1$, where $J$ is the largest nearest-neighbour bond norm and $\alpha$ is the spectral half-width of the physical Hamiltonian. On-site terms only generate single-site rotations, which an interaction picture absorbs, so they do not enter $v_\LR$. We use no extra padding, $\ell=0$, and define
\begin{equation}
    \begin{split}
        r_j(R^*)&\coloneqq\min\!\left\{R^*,\,1+\left\lceil2v_\LR|t_j|\right\rceil\right\},\qquad j\ne0,\\
        K_a^{[R^*]}&\coloneqq c_0^\star A_a+\sum_{j\ne0}c_j^\star e^{iH_{B_{r_j(R^*)}(S_a)}t_j}A_a e^{-iH_{B_{r_j(R^*)}(S_a)}t_j},\\
        \mathcal L^{[R^*]}&\coloneqq-i[H,\cdot]+\sum_a\mathcal D[K_a^{[R^*]}].
    \end{split}
    \label{eq:cappedpatchmain}
\end{equation}
Only the jumps are truncated; the coherent part of the generator keeps the full Hamiltonian. At zero padding these are the jumps $K^{(\bs R^*)}_{a,N}$ and generator $\calL_{N,\bs R^*}$ of the capped schedule of \cref{eq:cap}. Distances count physical sites, with the two spin orbitals of a Hubbard site counted as one. Because the cap cuts the light-cone schedule short, the relevant bound is \cref{eq:espmain} evaluated at these radii. At $R^*=3$ the eight-site Ising chain uses radius $2$ for its shortest branches, $j=\pm1$, and radius $3$ for all the others; in every other case shown, every nonzero-time branch reaches the cap.

\begin{figure*}[t]
    \includegraphics[width=\textwidth]{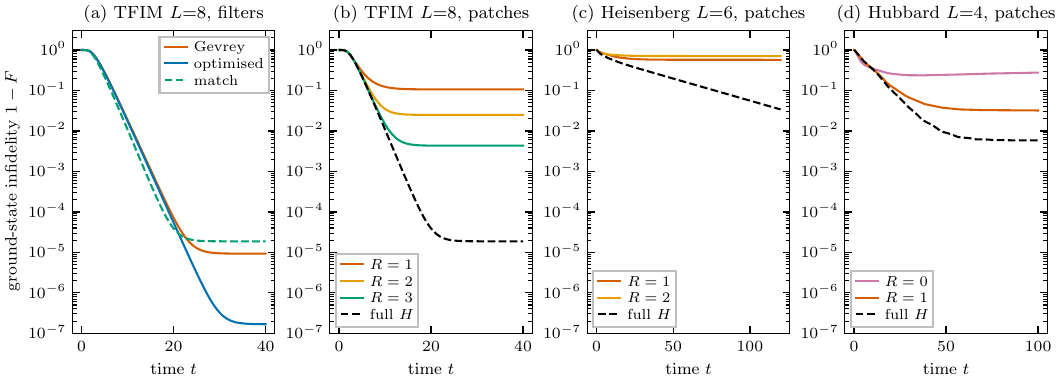}
    \caption{Convergence of the Lindblad dynamics from a zero-overlap initial state. (a)~Eight-site transverse-field Ising chain: ground-state infidelity against time for the Gevrey filter, the optimised filter at equal cost, and the match filter at matched design-grid leakage, all with full-Hamiltonian branches. The curves show comparable early relaxation and distinct late-time plateaux; the match cost is \MatchCostTFIMeight{} of the reference. (b)~The match filter with light-cone radii $r_j(R^*)=\min\{R^*,1+\lceil2v_\LR|t_j|\rceil\}$ around each coupling, with $v_\LR=2\mathrm eJ/\alpha$, bond norm $J$ and spectral half-width $\alpha$; the light cone of the longest branch spans the whole chain, so every patch shown truncates it. (c)~and (d)~The same for the six-site Heisenberg chain at $h = 0.5$ and the four-site Hubbard chain. \Cref{tab:lattice} gives stationary populations and the eight-site values at $t=40$.}
    \label{fig:convergence}
\end{figure*}

\Cref{fig:convergence}(b)--(d) show the capped dynamics $\mathcal L^{[R^*]}$ of \cref{eq:cappedpatchmain} for $R^*$ from $0$ to $3$. Every patch here lies far inside the light cone of the longest branch, whose radius $R(T)$ exceeds the chain length, so the bound of \cref{eq:espmain} guarantees nothing useful for these tests. The Ising and Hubbard chains work well anyway. With $R^*=1$, a patch of three sites around a single-site coupling or four around a bond, the ground-state population is at least $0.893$ (at the final time for the eight-site chain) [\cref{tab:lattice}; \cref{fig:convergence}(b) and (d)], and each unit increase in $R^*$ cuts the remaining infidelity by a factor of four to ten. The Hubbard jumps reach \HubbardRzeroF{} even at $R^*=0$, on the bond alone. The Heisenberg chain at the smaller field, $h=0.5$, behaves differently: its ground-state population is only \HeisStrongRone{} at $R^*=1$ and \HeisStrongRtwo{} at $R^*=2$, falling as the patch grows [\cref{fig:convergence}(c)]. This is also the chain on which even the untruncated dynamics is slow. Its Lindbladian gap is \MixRatioHeis{} times smaller than that of the six-site Ising chain, and starting from the highest eigenstate it reaches a ground-state population of only \HeisFonetwentyX{} by time $120$, against a stationary value above $0.999$. These small systems give some evidence that truncated patches can work when the dynamics mixes quickly and the truncation is small enough.

Whether a patch works therefore depends on two things: how much error the truncation introduces, and how quickly the dynamics removes it. To measure the first, we define the rate at which the truncated jumps heat the ground state,
\begin{equation}
    \Gamma_\uparrow(R^*)\coloneqq\sum_a\bigl\|(\ident-\Pi_0)K_a^{[R^*]}|\psi_0\rangle\bigr\|^2.
    \label{eq:heatingdiagnosticmain}
\end{equation}
At $R^*=1$ this rate is about $0.2$ for the Ising and Heisenberg chains and almost a hundred times smaller for Hubbard. It falls geometrically with $R^*$ on the Ising chain, and more slowly on the smaller-field Heisenberg chain. It measures transitions out of a single state, the ground state, whereas \cref{thm:stationary} uses the induced generator norm, a worst case over all inputs, together with a mixing estimate and its prefactor. For the second, the measured Lindbladian gaps of the smaller-field Heisenberg and Hubbard chains are \MixRatioHeis{} and \MixRatioHub{} times smaller than that of the six-site Ising chain. Fast relaxation is not enough on its own either: some Heisenberg patch generators relax faster, but to a state with a lower ground-state population. Both the patch error and the relaxation rate are therefore worth measuring when choosing a patch.

\subsection{Square Lattices: Grading the Radii}

\begin{table}[b]
    \caption{Ising model at $N=41$, $T=6\pi$, $R^*=64$, with coefficient budget $\Lambda_{\max}=1.25\norm{\bs d}_1$, where $\norm{\bs d}_1$ is the Gevrey one-norm. C denotes the compact ladder. U uses radius $64$ throughout; G uses graded radii; S grades radii after shortening the support. $N/N_{\rm ret}$ gives the candidate and retained label counts (C--U and C--G deliberately retain the prescribed $41$-label schedule, including its two zero-amplitude endpoints, to isolate grading; it is not the filter's active support, whose largest magnitude is $19$, and a baseline recompiled to that support would cost less), $R^*/R_{\rm ret}$ the radius cap and the largest radius used, and $\bar R$ the mean radius over evolution pieces. $\norm{\bs{c}^\star}_1$ is the coefficient one-norm and $\eup$ the raising error, bounded by the certified whole-band leakage because the couplings are assumed to have zero ground-state expectation. The error bound is the conditional trace-norm bound of \cref{eq:statescoremain}, using mixing margin $g_d=0.50$. TFIM cost is the ladder cost $C_{\rm BE}$ with the Ising cost functions above; cost change is $100(C/C_{\rm C-U}-1)\%$. The three rows isolate radius grading and support shortening within that ladder.}

    \label{tab:ablation}
    \begin{ruledtabular}
        \begin{tabular}{lcccccccc}
            design & $N/N_{\rm ret}$ & $R^*/R_{\rm ret}$ & $\bar{R}$ & $\norm{\bs{c}^\star}_1$ & $\eup$ & error bound & TFIM cost & cost change \\
            \colrule
            \TabRowLU \\
            \TabRowLG \\
            \TabRowLS \\
        \end{tabular}
    \end{ruledtabular}
\end{table}

Grading only pays off when the light cone is much smaller than the system, so we test it on square lattices, using cost functions built from the term counts of square-lattice Ising and Heisenberg blocks. For a patch of side $q=2R+1$, the number of Pauli terms $L_t(R)$ and their coefficient one-norm $\lambda_H(R)$ are
\begin{align}
    L_{t}^{\rm I}(R)&=2q(q-1)+2q^2,&
    \lambda_{H}^{\rm I}(R)&=2q(q-1)+2.5q^2,\nonumber\\
    L_{t}^{\rm H}(R)&=6q(q-1)+q^2,&
    \lambda_{H}^{\rm H}(R)&=6q(q-1)+9q^2,
    \label{eq:scoreblocksmain}
\end{align}
and the costs of preparing their coefficients are $C_{\rm prep}^{\rm I}=\operatorname{round}(0.99L_t^{\rm I}+25)$ and $C_{\rm prep}^{\rm H}=\operatorname{round}(0.69L_t^{\rm H}+136)$, rounded to the nearest integer. We combine these with a filter band $[\Delta,W]=[0.2,2]$ and locality parameters $C_4=\mu=v_\LR=1$, all chosen independently, so that the estimates isolate the effects of support reduction and grading. A calculation for a physical Hamiltonian would instead fix all of these from a single energy scale, through \cref{eq:energyscalingmain}.

The error-bound column of \cref{tab:ablation} evaluates \cref{eq:statescoremain} with implementation allowance $\zeta_{\rm imp}=10^{-3}+N\norm{\bs{c}^\star}_1\,2^{-16}$, assuming every coupling has zero ground-state expectation. The coefficient budget is $\Lambda_{\max}=1.25\norm{\bs d}_1$, the displacement cap uses $\gamma=0.8$ in \cref{eq:linearcapmain} with mixing margin $g_d=0.50$, and the support stage may exceed the full-support optimum's leakage and displacement by $5\%$ ($\xi=0.05$), within the same cap. The bound holds only under these assumptions. Each row uses full-precision coefficient vectors.

\Cref{tab:ablation} keeps the full-support coefficients fixed between the C--U and C--G rows, so their difference is due to grading alone. Grading the five ladder digits lowers the cost from $1.11\times10^{12}$ to $5.14\times10^{11}$ Toffolis, a saving of $53.8\%$; the cheapest complete weight list is $(7,6,4,1,2)$, at radii $(40,51,59,61,64)$ with mean $55.0$. Shortening the support to $35$ active branches lowers the graded cost further, to $3.18\times10^{11}$, a factor of $1.62$ below C--G and $3.5$ below C--U. The error bounds of the full and shortened graded designs are $0.04044$ and $0.04208$; shortening changes the raising error, the coefficient norm, the remaining margin $g_d-\greq$ and the implementation allowance, and so moves the bound slightly.

\begin{table*}[b]
    \caption{Outer-loop records on the open Ising chains at transverse field $1.2$, from the highest-energy eigenstate. (a)~First sampled time $t_\varepsilon$ with $1 - F(t_\varepsilon) \le \varepsilon$, samples at $t = 0, 0.1, \ldots, 400$; range and spread $(t_{\max} - t_{\min})/t_{\min}$ include only the filters reaching the target. (b)~Jump order: the six-site chain split into single-jump channels of unit duration with generator $\calL_a = -i[H,\cdot]/6 + \calD[K_a]$; ten sweeps, total jump time $60$, against the unsplit generator at time $10$; the ratio is $(1-F)/(1-F_{\mathrm{unsplit}})$ at these matched times. The two random orders are single realisations, with no ensemble uncertainty estimated; independent sampling applies the six jumps $(11,14,8,9,10,8)$ times, where the sweeps apply every jump ten times.}
    \label{tab:outerloop}
    \begin{ruledtabular}
    \small
        \begin{tabular}{lcccc@{\hspace{2em}}lcc}
            \multicolumn{5}{l}{(a) Sampled hitting times} & \multicolumn{3}{l}{(b) Jump order, six sites} \\
            chain & $\varepsilon$ & reaching it & time range & spread & order & $1-F$, ten sweeps & ratio \\
            \colrule
            TFIM, $L=4$ & $10^{-2}$ & $7/7$ & $9.9$--$11.9$ & $20.2\%$ & left to right & $0.016374$ & $1.053$ \\
            TFIM, $L=4$ & $10^{-3}$ & $6/7$ & $14.2$--$15.0$ & $5.6\%$ & alternating direction & $0.016308$ & $1.049$ \\
            TFIM, $L=4$ & $10^{-4}$ & $5/7$ & $18.9$--$19.8$ & $4.8\%$ & random permutation & $0.016792$ & $1.080$ \\
            TFIM, $L=6$ & $10^{-2}$ & $7/7$ & $10.2$--$11.6$ & $13.7\%$ & independent random site & $0.030644$ & $1.970$ \\
            TFIM, $L=6$ & $10^{-3}$ & $6/7$ & $14.2$--$15.3$ & $7.7\%$ & unsplit generator & $0.015553$ & $1$ \\
            TFIM, $L=6$ & $10^{-4}$ & $5/7$ & $18.6$--$19.4$ & $4.3\%$ & & & \\
        \end{tabular}
    \end{ruledtabular}
\end{table*}

\subsection{The Outer Loop}
The last question is how many times the outer loop must call the jumps. On the four- and six-site Ising chains we measure three things across one family of filters: how the relaxation rate depends on the filter, what sets the stationary floor, and what the order of the jumps costs (\cref{tab:outerloop}). The family cuts the reference grid off at a maximum time $T_{\rm target}$, with $T_{\rm target}\Delta$ from $1$ to $5$, and re-optimises the coefficients for each cut. We compare two ratios of heating to cooling. The band ratio $r_{\rm band}\coloneqq L_\Delta[h_N]/h_N(-\Delta)$ divides the worst heating anywhere on the band by the cooling at the band edge; the gap ratio $r_{\rm gap}\coloneqq|h_N(\Delta)/h_N(-\Delta)|$ compares heating and cooling at the gap itself. The two differ for an oscillating filter, whose worst heating need not sit at the gap.

The relaxation rate hardly changes while the band ratio varies over three decades. The Lindbladian gap varies by only about a tenth, and the sampled times to reach a target infidelity spread by about $20\%$ at $10^{-2}$, $8\%$ at $10^{-3}$ and $5\%$ at $10^{-4}$. Part of that narrowing comes from filters dropping out: the number of filters that reach each target falls from seven to five as the target tightens, and the filters lost are those with the highest plateaux. Among the filters that reach every target, the spread is about six per cent or less throughout.
The stationary floor, by contrast, follows the gap ratio. A two-level system with raising-to-cooling amplitude ratio $r_{\rm gap}$ has stationary excited population $r_{\rm gap}^2/(1+r_{\rm gap}^2)\approx r_{\rm gap}^2$ (\cref{app:obs1}). Across the fourteen filters, whose stationary infidelities span $6.3$ decades on four sites and $6.9$ on six, $(1-F_\infty)/r_{\rm gap}^2$ stays between $0.85$ and $1.37$ on four sites and between $0.78$ and $1.33$ on six, within $37\%$ of this simple prediction.
Unlike the continuous generator, the outer loop must also choose the order in which it applies the individual jumps. We split the six-site chain into single-jump channels of unit duration and compare orderings at equal total jump time. After ten sweeps, applying the jumps left to right leaves an infidelity of $1.64\times10^{-2}$, and alternating the direction on each pass leaves $1.63\times10^{-2}$, against $1.56\times10^{-2}$ for the unsplit generator. A fresh random order on each sweep leaves $1.68\times10^{-2}$, but drawing an independent random site at every step leaves $3.06\times10^{-2}$, almost twice the unsplit value.
\section{Discussion}
\label{sec:outro}

Dissipative approaches to ground-state preparation trade the need for a good initial state, present for coherent methods such as quantum phase estimation, for costly Lindblad dynamics whose stationary state is close to the ground state. Recent work reduced their ancilla requirement to a single qubit~\cite{DingChenLin2024}, which makes them initially seem attractive for early fault-tolerant machines; however, their high non-Clifford gate counts and run times temper near-term optimism. In this paper we have presented a construction for the main ingredient of dissipative techniques, the filtered jump operator, which can, under the right circumstances, significantly reduce the cost of implementation. Our construction is based on three observations, which we now discuss in turn.

The first observation is that designing the filter in the time domain turns filter design into linear programming. Once the evolution times are fixed and the coefficients are paired, the response is linear in the amplitudes, so the best filter in this family can be found exactly, and the cheapest acceptable support is found by searching over a single integer, the longest retained time. On every model and molecule we studied, the re-optimised filter reaches the Gevrey filter's own leakage using less than half of its maximum time, so the longer branches can simply be removed. This conclusion is tied to the Gevrey time grid, whose nodes we kept. Optimising the times as well would put them inside the cosines of the response and lose the linear program, although it would not rule out a global optimisation by other means. There is also no reason to expect a uniform grid to be best. The support stage removes the longest branches first, and a natural next step is to optimise the nodes continuously, keeping a few long branches for the sharp edge at the gap and many short, cheap ones for the rest of the response.

We also treated every coefficient as a free parameter. A basis designed for signals limited in both time and frequency, such as the discrete prolate spheroidal sequences, would restrict the family, so the best leakage it could reach would be no lower than ours. It would bring two benefits in return. Its long-time coefficients are small, and those coefficients dominate the second moment $\sum_j|c_j|t_j^2$ that controls the error between sampled energies, so the whole-band certificate would tighten. Its effective dimension, roughly the product of the time span and the bandwidth, would also predict the support length that the support stage currently finds by search.

Finally, we fixed what the filter has to achieve: the heating band $[\Delta,W]$ and the cooling floor $P(e)$ are set by the preparation target, and we asked how cheaply they can be met. Tasks that need less can ask for less. \citet{LiYangLin2026phase}, for example, decide which phase a system is in from the early-time response of phase-sensitive observables under filters with coarse energy resolution, and reach a decision long before the dynamics has mixed. Their saving comes from asking less of the filter and ours from meeting a fixed request more cheaply, so the two are independent, and a phase-decision protocol built on a time-domain filter like ours would benefit from both.

The second observation is that a light cone is set by elapsed time alone. In the ladder, every branch has evolved for at most $u\tau$ by the end of the digit at cumulative weight $u$, so that digit can be generated on the patch this time requires, which for the early digits is much smaller than the patch the longest branch needs. Giving each digit the radius of its elapsed time keeps the light-cone error bound of \cref{eq:adaptivemain}, because the patch boundary stays outside the light cone at every step. Generating each branch on a single patch sized for its final time gives a sharper bound, which the graded digits need not meet, since how much evolution remains after a digit depends on the branch. At the design point of \cref{tab:ablation}, grading lowers the cost by $53.8\%$, and across the design family the saving grows with the range of evolution times in the schedule. The saving is a constant factor, independent of system size, set by how much smaller the early patches are than the full light cone. A smaller Lieb--Robinson velocity shrinks every radius, but whether it increases the relative saving depends on what else is held fixed. Grading only helps when the light cone is much smaller than the system: in the chains and molecules of \cref{tab:lattice,tab:molecules}, $\lambda_H\tau$ is of order one and the light cone spans the whole system, so there is nothing to grade.

The third observation, that convergence survives localisation, holds where the dynamics mixes quickly and the truncation is small enough, and the mixing rate is an input to it: our stationary-state bound transfers a mixing estimate for the Gevrey generator and does not prove one. For weakly interacting systems, where cluster expansions already make the Gibbs state classically tractable at low~\cite{HelmuthMann2023}, high~\cite{MannMinko2024} and, more recently, arbitrary temperature~\cite{WaiteMann2026}, rapid mixing has been proved for a related one-sided filtered-jump construction~\cite{Zhan2025}; using that result here would need a transfer argument to our finite erf discretisation, which we have not carried out. Ding, Chen and Lin prove that the fixed point is unique when the couplings satisfy an ergodicity condition and show convergence numerically on Ising and Hubbard models~\cite{DingChenLin2024}, but they give no mixing pair for the chains studied here. Rates are starting to appear where the structure is favourable: a Heisenberg chain cooled at one end by an unfiltered lowering coupling has a Lindbladian gap that closes as $n^{-3}$~\cite{BeerBurgarth2026}, although that result concerns unfiltered relaxation and does not transfer directly to a filtered construction.

Under rapid mixing, the time that matters can be shorter than the time for the whole state to mix. Sums of geometrically local observables of a quasi-local Lindbladian equilibrate in a time independent of the system size~\cite{Smid2026}, which caps the number of outer-loop steps needed for such observables, and the same hypothesis gives a classical estimate of them in time linear in $n$. Both rely on rapid mixing, and neither covers the molecules here, whose observables are not geometrically local and for which no such certificate exists. For strongly interacting systems the rate has to be measured, as in our chain simulations, and the stopping rule of \cref{sec:theory} means it need not be known in advance.

The chain experiments also show what the bound does not. Patches that the bound cannot certify at all still work well on the Ising and Hubbard chains: the infidelity is $3$ to $11\%$ at $R^*=1$ and a few parts in a thousand at $R^*=3$ on the eight-site chain, and on these chains the accuracy improves with $R^*$ in every case we tested. The Heisenberg chains show that relaxing to a stationary state is a different thing from relaxing to the ground state. The diagnostic $\Gamma_\uparrow(R^*)$ of \cref{sec:numerics} tracks these outcomes, but it depends on the filter, the coupling, the Hamiltonian and the target as well as on the truncation, while the theorem bounds the shift of the stationary state by an induced generator norm over a mixing margin. Whether a bound that charges the patch error only at the ground state, such as $\Gamma_\uparrow$ over the Lindbladian gap, would be tight where the method works is an empirical question that the global bound leaves open. Fast mixing helps when the perturbation is small enough, and it does not replace that condition. The Hubbard chain adds a practical rule that the theorem does not state: if a coupling has a nonzero ground-state expectation, the filter's zero-frequency response acts as a Hamiltonian perturbation that can shift the stationary state linearly, and the fix is to centre the coupling, since suppressing that response in the filter costs too much leakage.

None of the example systems is a candidate for quantum advantage: each was chosen so that the dissipative dynamics could be checked exactly. The chains have at most eight sites and are simulated here classically; at any length, the transverse-field Ising chain maps to free fermions~\cite{LiebSchultzMattis1961}, the Heisenberg and Hubbard chains are solvable by Bethe ansatz~\cite{Bethe1931,LiebWu1968}, and gapped one-dimensional ground states can be found in polynomial time~\cite{LandauVaziraniVidick2015}. The molecular active spaces, of at most twenty qubits, are within reach of exact diagonalisation, and the ferromagnetic transverse-field Ising model on the square lattice, as on any graph, has a polynomial-time classical approximation of its partition function and ground-state energy~\cite{BravyiGosset2017}. The resource counts therefore price the primitive where its output can be verified. An advantage needs a system whose dissipative dynamics mixes quickly while cluster expansions~\cite{HelmuthMann2023,MannMinko2024,WaiteMann2026} fail to converge, quantum Monte Carlo has a sign problem and tensor networks lose accuracy, as in frustrated or doped two-dimensional models and strongly correlated molecules; for the latter, the evidence for an exponential advantage is still debated~\cite{Lee2023evidence}.

Our molecular results use only the first observation, and more could be done for molecules. In a delocalised orbital basis every branch simulates the whole Hamiltonian. We tested whether an orbital patch could play the role of a lattice patch on a hydrogen chain in a Boys-localised minimal basis: restricting the Hamiltonian that generates a central hopping coupling to its four nearest orbitals, out of six, gives a branch error of $0.13$ at normalised time $2$ and $0.9$ at time $10$. About half of the discarded one-norm is the Coulomb tail $n_i n_{i'}$, which is neither finite-range nor of bounded degree, so it falls outside the locality assumptions used here. The block structure of tensor-hypercontracted Hamiltonians~\cite{Lee2021}, or the energy locality of a low-energy subspace, could take the place that geometric locality plays on a lattice.

Working in a low-energy subspace also brings prior knowledge of the state. A dissipative continuation that follows a ground state along a reaction coordinate, checking the energy certificate at each point, would be another route to the transition-state searches that nudged-elastic-band methods perform classically; we have not pursued it. Nor have we tried a schedule of filters keyed to the energy certificate, which would spend the expensive long branches only in the final approach to the ground state. The difficulty is that the certificate bounds only a mean energy and a ground-state population, so the weight above an intermediate threshold $E_0+\Delta'$, with $\Delta'>\Delta$, is bounded only by $(\Tr[H\rho]-E_0)/\Delta'$, and a later filter designed for a narrower band must carry that tail, and its dynamics, in its error budget. Degenerate ground spaces would also need separate control of populations and coherences, which the present theorem does not provide.

Our counts rest on a simple cost model, and that model is where further savings lie. The query count is a first-order estimate, the block-encoding cost is a generic linear-combination-of-unitaries count, and each patch is charged as its enclosing box. Structured block encodings of lattice Hamiltonians~\cite{Babbush2018,LowChuang2019} and product formulas that exploit geometric commutation~\cite{ChildsSu2019} cost less than the generic count, and the short branches are a natural place to use them. A short Trotterised patch evolution is a product of small-angle rotations, whose cost falls with the angle in the probabilistic and quasiprobabilistic channel syntheses of Ref.~\cite{Bothe2026smallangle}. Those syntheses produce channels, though, and before their counts can be used one must specify how their error combines with a coherently controlled branch register, where \cref{eq:implementationmain} needs a unitary $\widetilde V_j$. A product formula applied directly is unitary and enters as an ordinary conjugator error, so this obstacle is specific to the quasiprobabilistic route. Conveniently, the short branches, which need the fewest Trotter steps, also use the smallest patches.

Jumps on disjoint patches in one Lindblad step can be applied in parallel, which leaves the gate count unchanged but divides the depth by the number of patches that fit. Product formulas with nearly linear spacetime cost~\cite{ChildsSu2019} could also change the quadratic dependence on patch volume assumed here. Of our three observations, only the second changes how the cost scales with system size under this cost model, from the whole lattice in every branch to a patch; the first and third change constants. It is in these constants, though, that the cost of ground-state preparation on early fault-tolerant hardware will be decided. The choices that set them, namely the filter, its time nodes, the compilation of the jump and the radius on which each piece is generated, are ones every implementation has to make, and their effect on the constants deserves the same scrutiny that the exponents have received.

\begin{acknowledgments}
We thank Adrian Chapman, Simon Devitt, Ryan Mann, Jannis Ruh, Gabriel Waite, Thomas Watts and Nathan Wiebe for thought-provoking discussions.
Project led by University of Technology Sydney and supported by Defence Science and Technology Group (DSTG) and Advanced Strategic Capabilities Accelerator (ASCA) through its Emerging and Disruptive Technologies (EDT) Program.
\end{acknowledgments}

\bibliography{refs}

\appendix

\section{Setting}
The notation and working space are fixed in \cref{sec:theory}. The locality proofs require a Lieb--Robinson bound that holds uniformly over restricted Hamiltonians, including commutators with terms that straddle the boundary of the restricted region.

\begin{assumption}[Uniform Lieb--Robinson bound; standard form from Refs.~\cite{LiebRobinson1972,HastingsKoma2006,NachtergaeleSims2006}]
    \label{ass:LR}
    There are constants $v_\LR > 0$, $\mu > 0$, finite $C_\LR$ such that for every $Y \subseteq \mathcal{V}$, operator $O_X$ supported on $X \subseteq Y$, and $O_Z$ supported on an arbitrary $Z \subseteq \mathcal{V}$, and every $t \in \mathbb{R}$,
    \begin{equation}
        \bigl\lVert \bigl[ e^{iH_Y t} O_X e^{-iH_Y t},\, O_Z \bigr] \bigr\rVert
        \;\le\; C_\LR\, \norm{O_X}\, \norm{O_Z}\, \min\{|X|,|Z|\}\, e^{-\mu\left(\dist(X,Z) - v_\LR |t|\right)}.
        \label{eq:LR}
    \end{equation}
\end{assumption}

The same constants apply to every region $Y$, so $O_Z$ can be a term of the full Hamiltonian crossing the boundary of $B_R(S_a)$. For the interactions considered here, with bounded degree, hyperedge size and term norms, the standard bounds give $v_\LR$ of order $Jk\Delta_G$, with constants depending only on $(k,\Delta_G,J,V)$. The sum over distance shells in the proof of \cref{prop:locality-formal} uses subexponential growth, or $V(r)\le C_\nu e^{\nu r}$ with $0<\nu\le\mu/2$. For more detailed exposition see Refs.~\cite{ChenLucasYin2023,WangHazzard2020}.

\begin{assumption}[Spectral gap and nondegenerate ground state; standing model assumption]
    \label{ass:gap}
    As in \cref{sec:theory}, $H$ has a unique ground state $\ket{\psi_0}$ on the working space, spectral gap $\Delta_{\rm spec} \coloneqq E_1 - E_0 > 0$ and design gap $0 < \Delta \le \Delta_{\rm spec}$.
\end{assumption}

This assumption fixes the single target used in the stationary-state bounds. The energy certificate \cref{eq:energycertmain} also applies to a degenerate low-energy projector (\cref{eq:energycert}).

\section{Observation 1: finite filters optimised in the time domain}
\label[appendix]{app:obs1}

This appendix proves the two-stage optimisation result and the same-resource certificate behind the first observation.
The heating and cooling constraints have a simple physical interpretation. To illustrate, consider a two-level mode of energy $e>0$ and let $b_e=\ket{0}\bra{1}$ and $n_e=\Tr[b_e^\dagger b_e \rho]$. The filtered jump $h_N(-e)b_e+h_N(e)b_e^\dagger$, at unit rate, has cross terms in the dissipator that act on coherences and vanish on the number observable. Thus
\begin{equation}
    \dot{n}_e = -\bigl(|h_N(-e)|^2 + |h_N(e)|^2\bigr) n_e + |h_N(e)|^2,
    \qquad
    n_e(\infty) = \frac{|h_N(e)|^2}{|h_N(e)|^2 + |h_N(-e)|^2}.
    \label{eq:mode}
\end{equation}
Minimising the heating leakage $L_\Delta[h_N]$ of \cref{sec:theory} while holding the cooling floor $h_N(-e) \ge P(e) > 0$ on the same band controls this stationary occupation.

\begin{proof}[Proof of \cref{thm:twostage}]
    On the fixed symmetric time grid, $d_{-j}=d_j^*$ and the real amplitudes obey $x_{-j}=x_j$. Hence \cref{eq:mfamily} is real and affine in $\bs x$. With an auxiliary variable $\varepsilon$, the sampled heating constraint becomes $-\varepsilon\le h_N(e)\le\varepsilon$ at each $e\in\mathcal E$, and minimising $\varepsilon$ minimises the sampled leakage; the cooling floor is affine. Splitting each $x_j$ into its positive and negative parts makes $\norm{\bs c}_1=\sum_j|d_j||x_j|\le\Lambda_{\max}$ linear, giving a linear program~\cite{BoydVandenberghe2004}. After reference coordinates with negligible amplitude are fixed to zero as in \cref{sec:theory}, every remaining $d_j$ is nonzero. The coefficient budget therefore bounds every amplitude, and its intersection with the closed band constraints is compact. The assumed nonempty feasible set and continuity of the sampled objective give a global minimiser.
    
    For the support stage, the coefficient budget gives $|d_j||x_j|\le\Lambda_{\max}$, so $M_j=\Lambda_{\max}/|d_j|$ is a valid bound in $|x_j|\le M_jz_j$, with paired coordinates sharing $x_j$ and $z_j$. The displacement cap is linear after the same splitting. Each fixed binary vector leaves a closed subset of the same compact coefficient set, and there are finitely many binary vectors. Since $\eta_{\sup}\ge\eta(\bs c_{\rm spec},\bs d)$, the first-stage minimiser is feasible with its nonzero branches active for every $\xi\ge0$. A support therefore attains the least cost. The resulting program is mixed-integer linear~\cite{Wolsey1998}; every admissible coefficient vector can be represented with exactly its active-branch cost $\sum_jW_jz_j$. Its optimum is global within the stated family and sampled constraints.
\end{proof}

Whole-band feasibility follows from a derivative bound. For any coefficient vector $\bs q$, differentiating \cref{eq:mjump} twice gives
\begin{equation}
    |h_N''(e;\bs q)|\le M_2(\bs q)\coloneqq\sum_j|q_j|t_j^2,
    \qquad
    |h_N(e;\bs q)-\mathcal I_h[h_N](e;\bs q)|\le\frac{M_2(\bs q)h^2}{8},
    \label{eq:interpolation}
\end{equation}
where $\mathcal I_h[h_N]$ is the linear interpolant on an interval of width $h$. For the real paired response, the interpolation remainder bounds the heating magnitude above the larger endpoint magnitude. Applied to $\bs q=\bs c-\bs d$, it bounds the cooling gain below the smaller endpoint gain. These inequalities certify feasibility between samples; the global-optimality statement of \cref{thm:twostage} concerns the sampled program.

The certificate in \cref{fig:certificate} uses $\norm{H}=1$, $\Delta=0.2$, $W=2$, and the $81$ trapezoidal nodes $t_j=j\pi/5$, $|j|\le40$, so $T=8\pi$. The Gevrey vector from \cref{eq:window} has $\norm{\bs d}_1=1.4709191989\ldots$ and $h_{\eref,N}(0.2)=0.0023073784\ldots$. Reweighting keeps this norm budget and imposes the sampled floor $h_{\eref,N}(-e)+2\times10^{-4}$. The solver reports objective $1.978\times10^{-5}$; the returned vector has sampled heating maximum $1.9863\times10^{-5}$ and minimum cooling gain $1.99974\times10^{-4}$ within the feasibility tolerance.

On $20{,}000$ equally spaced verification energies, $h=1.8/19{,}999$. The returned vector has $M_2(\bs c)\approx16.68$ and $M_2(\bs c-\bs d)\approx16.06$, giving interpolation remainders of approximately $1.7\times10^{-8}$ and $1.6\times10^{-8}$. Together with the sampled extrema, these give a whole-band cooling gain above $1.9995\times10^{-4}$ and
\begin{equation}
    \sup_{e\in[0.2,2]}|h_{\rm opt}(e)|<1.9881\times10^{-5},
    \qquad
    |h_{\eref,N}(0.2)|>2.3073\times10^{-3}.
    \label{eq:certnumbers}
\end{equation}
The sampled extrema and their remainders give a leakage reduction of $116.06$ at the same nodes, maximum time and norm budget. The certified vector has displacement about $1.02$.

The perturbation argument uses a mixing pair $(\lambda,\kappa)$, with $\lambda>0$ and $\kappa\ge1$, on a common invariant subspace of the working space:
\begin{equation}
\tnorm{e^{\calL t}X}\le\kappa e^{-\lambda t}\tnorm X
\quad(t\ge0)
\label{eq:mixingpair}
\end{equation}
for every traceless Hermitian $X$ supported on that subspace. Every generator compared in \cref{prop:transfer,thm:stationary} preserves it; uniqueness and convergence are asserted there.

\begin{proof}[Proof of \cref{prop:transfer}]
Unitary invariance and $\norm{A_a}\le1$ give
\begin{equation}
K_a(\bs c)-K_a(\bs d)=\sum_j(c_j-d_j)e^{iHt_j}A_ae^{-iHt_j},
\qquad
\norm{K_a(\bs c)-K_a(\bs d)}\le\eta(\bs c,\bs d).
\label{eq:jumpdist}
\end{equation}
Combining this with \cref{eq:dissdiff} and $\norm{K_a(\bs q)}\le\norm{\bs q}_1$ proves the generator bound \cref{eq:gendist}. Let $u(t)=\tnorm{e^{\calL(\bs c)t}X}$. Duhamel's formula, in the generator-comparison form of Ref.~\cite{Cubitt2015}, and the reference mixing pair give
\begin{equation}
u(t)\le\kappa_de^{-\lambda_dt}\tnorm X
+\kappa_d\delta_{c,d}\int_0^t e^{-\lambda_d(t-s)}u(s)\,\mathrm ds.
\label{eq:gronwall}
\end{equation}
Gr\"onwall's inequality gives the positive rate $\lambda_d-\kappa_d\delta_{c,d}$ under the stated hypothesis, and uniqueness follows from contraction. On the traceless Hermitian subspace, the inverse $-\int_0^\infty e^{\calL(\bs d)t}\,\mathrm dt$ of $\calL(\bs d)$ has norm at most $\kappa_d/\lambda_d$. Applying it to $\calL(\bs d)(\sigma_c-\sigma_d)=-[\calL(\bs c)-\calL(\bs d)]\sigma_c$ gives $\tnorm{\sigma_c-\sigma_d}\le\kappa_d\delta_{c,d}/\lambda_d$, completing \cref{eq:ratetransfer}.
\end{proof}

The mixing margin and linear displacement budget are \cref{eq:margin,eq:linearcapmain}. The numerical choice $g_d=0.5$ is an assumed scenario, equivalent to $\lambda_d/\kappa_d=0.5m$; the models have no proved mixing estimate establishing that value.

\section{Observation 2: light-cone patches}
\label[appendix]{app:obs2}
\label[appendix]{sec:locality}

Adapting the neighbourhood construction of Ref.~\cite{Hahn2026}, which fixes one radius for the whole jump, we assign branch $j$ its own nonnegative integer radius $R_j$ about the coupling support $S_a=B_{r_S}(a)$. The branch Hamiltonian contains the terms wholly inside $B_{R_j}(S_a)$, and the localised jump uses the selected coefficients $c_j^\star$:
\begin{equation}
    H_{a,R_j} \;\coloneqq\; \sum_{X \subseteq B_{R_j}(S_a)} h_X,
    \qquad
    K^{(\bs{R})}_{a,N} \;\coloneqq\; \sum_{j} c_j^\star\, e^{iH_{a,R_j} t_j} A_a e^{-iH_{a,R_j} t_j},
    \label{eq:localjump}
\end{equation}
where $\bs{R}=(R_j)$ and $R=\max_jR_j$. Replacing each jump in \cref{eq:lindblad_K} gives $\calL_{N,\bs{R}}$, with coherent evolution under the full $H$. Both jumps have norm at most $\norm{\bs{c}^\star}_1$.

\begin{proposition}[Branch-adaptive locality; formal version of \cref{prop:locality}]
\label{prop:locality-formal}
Assume \cref{ass:LR} and growth $V(r+r_S)\le C_\nu e^{\nu r}$ for some $0<\nu\le\mu/2$. For balls $B_{R_j}(S_a)$ around coupling supports $S_a=B_{r_S}(a)$, there is a constant $C_4>0$, depending only on $(k,\Delta_G,J,V,r_S)$ and the fixed locality constants, such that, for any radius vector $\bs R$, the \emph{spatial error}
\begin{equation}
\esp(\bs{R}; \bs{t}, \bs{c}^\star) \;\coloneqq\; \sum_{j} \bigl| c_j^\star \bigr| \min\Bigl\{ 2,\; C_4 |t_j| e^{\mu v_\LR |t_j|} e^{-\frac{\mu}{2}(R_j - 1)} \Bigr\}
\label{eq:espdef}
\end{equation}
obeys
\begin{equation}
\bigl\lVert K_{a,N} - K^{(\bs{R})}_{a,N} \bigr\rVert \;\le\; \esp(\bs{R}; \bs{t}, \bs{c}^\star),
\qquad
\delta_{N,\bs{R}} \;\coloneqq\; \indnorm{\calL_{N,\bs{R}} - \calL_N} \;\le\; 4 m \norm{\bs{c}^\star}_1\, \esp,
\label{eq:jumperror}
\end{equation}
where $\indnorm{\cdot}$ is the induced trace norm. The uniform choice $R_j = R$ is bounded by
\begin{equation}
\esp^{\mathrm{unif}}(R; T, \norm{\bs{c}^\star}_1) \;\le\; \norm{\bs{c}^\star}_1 \min\Bigl\{ 2,\; C_4 T e^{\mu v_\LR T} e^{-\frac{\mu}{2}(R - 1)} \Bigr\}.
\label{eq:unif}
\end{equation}
For padding $\ell \ge 0$, choose
\begin{equation}
R_j(\ell) \;\coloneqq\;
\begin{cases}
1, & t_j = 0,\\[2pt]
1 + \lceil 2 v_\LR |t_j| + \ell \rceil, & t_j \neq 0,
\end{cases}
\qquad
R \coloneqq 1 + \lceil 2 v_\LR T + \ell \rceil.
\label{eq:adaptive-formal}
\end{equation}
Writing $M_1(\bs{c}^\star, \bs{t}) \coloneqq \sum_j |c_j^\star|\,|t_j|$ for the first absolute time moment, this schedule gives
\begin{equation}
\esp(\bs{R}(\ell); \bs{t}, \bs{c}^\star) \;\le\; C_4\, M_1(\bs{c}^\star, \bs{t})\, e^{-\mu \ell / 2}.
\label{eq:adaptivebound}
\end{equation}
\end{proposition}

\begin{proof}[Proof of \cref{prop:locality-formal}]
Set $A_a(t)=e^{iHt}A_ae^{-iHt}$ and $A_a^R(t)=e^{iH_{a,R}t}A_ae^{-iH_{a,R}t}$. The Heisenberg comparison of Ref.~\cite{HastingsKoma2006} gives, by Duhamel's formula,
\begin{equation}
\bigl\lVert A_a(t) - A^R_a(t) \bigr\rVert \;\le\; \int_0^{|t|} \sum_{X \nsubseteq B_R(S_a)} \bigl\lVert [h_X, A^R_a(\operatorname{sgn}(t)\,s)] \bigr\rVert \, \mathrm{d}s.
\label{eq:duhamel}
\end{equation}
The sum includes boundary-straddling terms. An omitted hyperedge has a vertex at distance at least $R+1$ from $S_a$, so $\dist(X,S_a)\ge R$ because its vertices are pairwise at distance one. For lattice distance and interaction range $r_0$, the lower bound is $R-r_0+1$, with the fixed shift absorbed into $C_4$. At distance $r$, a term has a vertex in $B_{r+r_S}(a)$ for fixed $a\in S_a$, giving at most $\Delta_G V(r+r_S)$ terms. Their norms and sizes obey $\norm{h_X}\le J$ and $|X|\le k$. The restricted Hamiltonian has the same bounds on degree, hyperedge size and term norm, so the uniform \cref{ass:LR} gives $\norm{[h_X,A_a^R(s)]}\le 2J\norm{A_a}C_\LR|X||S_a|e^{-\mu(r-v_\LR s)}$. For $V(r+r_S)\le C_\nu e^{\nu r}$ with $0<\nu\le\mu/2$, the shell sum decays at rate $\mu-\nu\ge\mu/2$; polynomial growth admits every $\nu>0$ with a suitable $C_\nu$. Integrating over $s\in[0,|t|]$ and using $\norm{A_a}\le1$ yields
\begin{equation}
\bigl\lVert A_a(t) - A^R_a(t) \bigr\rVert \;\le\; \min\Bigl\{ 2,\; C_4 |t| e^{\mu v_\LR |t|} e^{-\mu(R - 1)/2} \Bigr\},
\label{eq:branchtail}
\end{equation}
where the shift to $R-1$ covers the lattice convention, and $C_4$ depends only on $(k,\Delta_G,J,V,r_S)$ and the fixed locality constants. For faster growth, $\mu/2<\nu<\mu$, the tail becomes $e^{\mu v_\LR|t|-(\mu-\nu)(R-1)}$ and every radius schedule must use speed $\mu v_\LR/(\mu-\nu)$ in place of $2v_\LR$.

Weighting \cref{eq:branchtail} by $|c_j^\star|$ and summing proves the jump bound. The dissipator estimate \cref{eq:dissdiff} and the two jump norms bounded by $\norm{\bs{c}^\star}_1$ give the generator bound in \cref{eq:jumperror}. Taking $|t_j|\le T$ gives \cref{eq:unif}; the adaptive radius cancels $e^{\mu v_\LR|t_j|}$ branch by branch, giving \cref{eq:adaptivebound}.
\end{proof}

\begin{proposition}[Graded ladder digits]
\label{prop:walkradii}
    Under the assumptions of \cref{prop:locality-formal}, write $s_j=\operatorname{sgn}(t_j)$ for $t_j\ne0$. Compile branch $j$ as $V_jA_aV_j^\dagger$, where $V_j=U_{p_j}\cdots U_1$ applies $U_p=e^{is_jH_{a,\rho_p}\theta_p}$ in increasing $p$, with $\theta_p>0$ and $\sum_p\theta_p=|t_j|$. Let $u_p=\sum_{p'\le p}\theta_{p'}$ be the elapsed time through piece $p$. The zero-time branch uses the identity and contributes $c_0^\star A_a$. With the same constant $C_4$,
    \begin{equation}
        \bigl\lVert A_a(t_j) - V_j A_a V_j^\dagger \bigr\rVert
        \;\le\;
        \min\biggl\{ 2,\; \sum_{p} C_4\, \theta_p\, e^{\mu v_\LR u_p}\, e^{-\frac{\mu}{2}(\rho_p - 1)} \biggr\} .
        \label{eq:worddefect}
    \end{equation}
    Giving a digit of cumulative weight $u$ the radius $1+\lceil2v_\LR u\tau+\ell\rceil$ preserves \cref{eq:adaptivebound}, since every branch invoking it has elapsed time at most $u\tau$. Whether graded radii also meet the sharper bound for one uniform radius at the final elapsed time, and how their actual errors compare, must be checked separately, since the remaining evolution time depends on the branch.
\end{proposition}

\begin{proof}
    Set $B_p=A_a(s_ju_p)$ and $\tilde A_p=U_p\tilde A_{p-1}U_p^\dagger$, with $B_0=\tilde A_0=A_a$. Then $B_{p_j}=A_a(t_j)$ and $\tilde A_{p_j}=V_jA_aV_j^\dagger$. Inserting $B_p=e^{is_jH\theta_p}B_{p-1}e^{-is_jH\theta_p}$ and using unitary invariance gives
    \begin{equation}
        \bigl\lVert \tilde{A}_p - B_p \bigr\rVert
        \;\le\; \bigl\lVert \tilde{A}_{p-1} - B_{p-1} \bigr\rVert
        \;+\; \bigl\lVert U_p B_{p-1} U_p^\dagger - e^{i s_j H\theta_p} B_{p-1} e^{-i s_j H\theta_p} \bigr\rVert .
        \label{eq:steptelescope}
    \end{equation}
    This comparison preserves factor order and requires no commutation between distinct patch Hamiltonians. Duhamel's formula bounds the second term by $\int_0^{\theta_p}\sum_{X\nsubseteq B_{\rho_p}(S_a)}\lVert[h_X,A_a(s_j(u_p-s))]\rVert\,\mathrm ds$. The exact operator inside the commutator has evolved for time of magnitude at most $u_p$, and \cref{ass:LR} applies to both time signs. The shell sum used in \cref{eq:branchtail} therefore gives step error $C_4\theta_p e^{\mu v_\LR u_p}e^{-\mu(\rho_p-1)/2}$: elapsed time sets the exponential and step duration its prefactor. Telescoping and the norm bound $2$ prove \cref{eq:worddefect}.
    
    For a shared digit of cumulative weight $u$, $u_p\le u\tau$ gives $e^{\mu v_\LR u_p}e^{-\mu(\rho_p-1)/2}\le e^{-\mu\ell/2}$. Summing durations gives branch error $C_4|t_j|e^{-\mu\ell/2}$, and weighting by $|c_j^\star|$ proves \cref{eq:adaptivebound}.
\end{proof}

The C--G row of \cref{tab:ablation} uses this schedule.

\section{Observation 3: convergence on patches}
\label[appendix]{app:obs3}
\label[appendix]{sec:stability}
The argument below separates the Lieb--Robinson tails, which the truncation introduces, from the contraction supplied by the mixing, the same device used in Ref.~\cite{Smid2026} to bound how quickly local observables equilibrate. The statement differs: there the generator is exact and the truncation is an artefact of the proof, whereas here the generator itself is truncated, by design, and the quantity bounded is the global stationary state in trace norm rather than a local observable.

Assume the global generator has mixing pair $(\lambda_N,\kappa_N)$ as in \cref{eq:mixingpair}. Every compared generator preserves the same invariant subspace containing the initial and target states; uniqueness and convergence below are asserted on this space. The argument applies to any radius vector $\bs R$, including the capped schedule $\bs{R}^*$ of \cref{thm:stationary}.

\begin{proof}[Proof of \cref{thm:stationary}]
The mixing estimate makes $\calL_N$ invertible on traceless Hermitian operators, with inverse of norm at most $\kappa_N/\lambda_N$, since $\tnorm{\calL_N^{-1}X}\le\int_0^\infty\tnorm{e^{\calL_Nt}X}\,\mathrm dt\le(\kappa_N/\lambda_N)\tnorm{X}$. Denote its stationary state by $\sigma_N$; every jump has $\norm{K_{a,N}}\le\norm{\bs{c}^\star}_1$. For $v_a=K_{a,N}\ket{\psi_0}$ and $w_a=K_{a,N}^\dagger v_a$, direct expansion gives
\begin{equation}
\calD[K_{a,N}](\Pi_0) = \ket{v_a}\!\bra{v_a} - \tfrac{1}{2}\bigl( \ket{w_a}\!\bra{\psi_0} + \ket{\psi_0}\!\bra{w_a} \bigr),
\qquad
\tnorm{\calD[K_{a,N}](\Pi_0)} \le \eup^2 + \norm{\bs{c}^\star}_1 \eup.
\label{eq:heating}
\end{equation}
Since $[H,\Pi_0]=0$, summing over $m$ jumps and applying the inverse yields the first term in \cref{eq:mainbound}. Applying it to
\begin{equation}
\calL_N(\sigma_{N,\bs{R}} - \sigma_N) \;=\; -\bigl( \calL_{N,\bs{R}} - \calL_N \bigr)\, \sigma_{N,\bs{R}}
\label{eq:spatialterm}
\end{equation}
yields the spatial term. The Duhamel--Gr\"{o}nwall argument of \cref{prop:transfer} gives mixing and uniqueness when $\kappa_N\delta_{N,\bs R}<\lambda_N$. Finally,
\begin{equation}
K_{a,N}\ket{\psi_0} = \sum_k h_N(E_k - E_0)\, \ket{\psi_k}\!\bra{\psi_k} A_a \ket{\psi_0},
\label{eq:raising}
\end{equation}
so orthogonality gives $\norm{K_{a,N}\ket{\psi_0}}^2=\sum_k|h_N(E_k-E_0)|^2|\bra{\psi_k}A_a\ket{\psi_0}|^2$. The $k=0$ term is $|h_N(0)|^2|\bra{\psi_0}A_a\ket{\psi_0}|^2$; every other allowed transition lies in the band. Using $\norm{A_a}\le1$ gives \cref{eq:eupbound}. This zero-frequency term contributes to the residual: even when $v_a=\alpha\ket{\psi_0}$, \cref{eq:heating} can carry coherences through $w_a=\alpha K_{a,N}^\dagger\ket{\psi_0}$ at first order in $\alpha$.
\end{proof}

Coefficient rounding and finite-precision conjugators in the circuit of \cref{app:frontier} give the unnormalised jump error
\begin{equation}
    \bigl\lVert K_{a,N}^{(\bs{R})} - \widetilde{K}_{a,N}^{(\bs{R})} \bigr\rVert \;\le\; \sum_j |\widetilde c_j - c_j^\star| + 2\sum_j |c_j^\star|\,\norm{\widetilde V_j - V_j} \;\le\; \zeta_{\mathrm{imp}},
    \label{eq:zetaimp}
\end{equation}
with the allowance $\zeta_{\mathrm{imp}}$ of \cref{eq:implementationmain} chosen independently of the cost model's $\eHS$ and coefficient precision $b_{\mathrm{coef}}$. Its coefficient term follows from rounding; its conjugator term needs a Hamiltonian-simulation circuit whose accuracy is proven, including the polynomial approximation of each evolution and the synthesis of its rotations. The implemented norm can reach $\norm{\bs{c}^\star}_1+\zeta_{\mathrm{imp}}$, so \cref{eq:dissdiff} gives extra generator distance $m(4\norm{\bs{c}^\star}_1+2\zeta_{\mathrm{imp}})\zeta_{\mathrm{imp}}$. With $\greq=2[\norm{\bs{c}^\star}_1+\norm{\bs d}_1]\eta(\bs c^\star,\bs d)$, the remaining-margin condition $4\norm{\bs{c}^\star}_1(\esp+\zeta_{\mathrm{imp}})+2\zeta_{\mathrm{imp}}^2<g_d-\greq$ ensures that the generator with the implemented jumps mixes and its stationary state obeys
\begin{equation}
\tnorm{\widetilde\sigma - \Pi_0} \;\le\; \frac{\eup^2 + \norm{\bs{c}^\star}_1\, \eup + 4 \norm{\bs{c}^\star}_1\bigl[\esp(\bs{R}; \bs{t}, \bs{c}^\star) + \zeta_{\mathrm{imp}}\bigr] + 2\zeta_{\mathrm{imp}}^2}{g_d - \greq}.
\label{eq:compilerbound}
\end{equation}
\Cref{tab:ablation} evaluates this state error bound from each row's declared inputs. A finite-time run of the outer loop additionally incurs convergence and channel-splitting errors. The compiled rows assume couplings with zero ground-state expectation. Such couplings, or a justified centring procedure, must be supplied for the Hamiltonian blocks used; otherwise \cref{eq:eupbound} includes $|h_N(0)|\,|\bra{\psi_0}A_a\ket{\psi_0}|$.

Centring changes the generator unless its coherent contribution from \cref{eq:gauge} is compensated. The centred generator has the equivalent descriptions $-i[H,\cdot]+\calD[K+\beta\ident]=-i[H+H_\beta,\cdot]+\calD[K]$, hence the same stationary states in both. In the centred-jump description, $[H,\Pi_0]=0$ remains available, and the jump norm, target residual and mixing pair are evaluated for the centred generator. In the other description, the coherent residual $-i[H_\beta,\Pi_0]$ must be retained and can cancel the dissipator residual. Thus $[H_\beta,\Pi_0]\ne0$ alone does not establish nonstationarity. The theorem must be applied with the complete residual of a single description. The Hubbard comparison in \cref{sec:impact} measures the uncentred and centred generators; a bound for the centred dynamics would require its own mixing estimate and residuals.

For $\varepsilon\in(0,1]$, assume $\kappa_N\delta_{N,\bs R}\le\lambda_N/2$ and $\tnorm{\sigma_{N,\bs R}-\Pi_0}\le\varepsilon/2$ from \cref{thm:stationary}. Duhamel--Gr\"{o}nwall mixing then gives the run time
\begin{equation}
t_{\mathrm{stop}} \;=\; \frac{2}{\lambda_N} \log \frac{4\kappa_N}{\varepsilon}
\label{eq:tstop}
\end{equation}
for trace-norm error at most $\varepsilon$ and ground-state population at least $1-\varepsilon/2$. The energy test in \cref{eq:energycertmain} can certify a state without knowing this run time. For a pure target, $\tnorm{\rho-\Pi_0}\le2\sqrt{1-\Tr[\Pi_0\rho]}$; energy excess at most $\Delta\varepsilon$ therefore certifies population at least $1-\varepsilon$ and trace-norm error at most $2\sqrt\varepsilon$. Trace-norm tolerance $\varepsilon$ requires excess at most $\Delta\varepsilon^2/4$.

The population certificate also holds for a degenerate low-energy manifold. On the working space, with identity $\ident$, sector ground energy $E_0$ and spectral projector $\Pi_{<\Delta}$ onto $[E_0,E_0+\Delta)$, the inequality $H-E_0\ident\succeq\Delta(\ident-\Pi_{<\Delta})$ gives
\begin{equation}
1 - \Tr[\Pi_{<\Delta}\, \rho] \;\le\; \frac{\Tr[H\rho] - E_0}{\Delta},
\label{eq:energycert}
\end{equation}
which reduces to \cref{eq:energycertmain} under \cref{ass:gap}, where $\Pi_{<\Delta}=\Pi_0$. For the test to accept eventually, the energy and gap bounds must be accurate enough. Let $E(t_r)$ be the true energy at test $r$. When all the confidence bounds hold, the test accepts whenever the true excess $E(t_r)-E_0$, plus the overestimate $\overline E_r-E(t_r)$ and the underestimate $E_0-\underline E_0$, is at most $\varepsilon\underline\Delta$. Eventual acceptance follows if the stationary energy excess $\Tr[H\sigma_{N,\bs R}]-E_0$ lies below the uncertainty-adjusted threshold by a fixed positive margin for all sufficiently late tests. The stationary-state bound in \cref{eq:compilerbound} can control this excess under its stated hypotheses. If the stationary excess exceeds the threshold, convergence supplies no guarantee of acceptance.

\section{Fault-tolerant cost model and compiled resource estimates}
\label[appendix]{app:frontier}
\label[appendix]{sec:frontier}

The jump is block-encoded by the standard linear-combination-of-unitaries circuit~\cite{Babbush2018}. PREP loads the $N$ candidate labels with amplitudes $\sqrt{|c_j^\star|/\norm{\bs{c}^\star}_1}$; SELECT applies the phase and the conjugated coupling of the selected branch, skipping the simulation when $z_j=0$; and PREP$^\dagger$ undoes the preparation. One PREP--SELECT--PREP$^\dagger$ cycle encodes the jump divided by $\norm{\bs{c}^\star}_1$. Pauli couplings are unitary; a hopping coupling combines two Pauli strings per spin and needs its own preparation, charged separately.

\begin{table}[!b]
\caption{Per-jump block-encoding costs on the $N=41$, $T=6\pi$, $\Delta=0.2$ grid. Gevrey retains all $41$ labels; match keeps the shortest $\MatchActiveTwoD$ branches at the reference leakage, with maximum time $\MatchTeffTwoD$ and grading where the radius permits. The lower panel compares dense TFIM with match patches at $R^*=3$. ``Total'' multiplies the per-jump cost by $L^2$ jumps and $\StepsAssumed$ assumed applications. These are cost estimates for the primitive at fixed normalised gap; neither the number of applications nor the accuracy of two-dimensional patches has been checked by simulation.}
\label{tab:resources}
\begin{ruledtabular}
\small\setlength{\tabcolsep}{3pt}
\begin{tabular}{lccccc}
 & \multicolumn{2}{c}{TFIM} & \multicolumn{2}{c}{Heisenberg} & \\
\cline{2-3}\cline{4-5}
$R^*$ (patch qubits) & Gevrey & match & Gevrey & match & \\
\colrule
\ResTableRows
\colrule
\multicolumn{6}{l}{TFIM: full-lattice Gevrey versus match patches at $R^*=3$}\\
lattice & jump, Gevrey & jump, ours & total, Gevrey & total, ours & ratio \\
\ResTotalRows
\end{tabular}
\end{ruledtabular}
\end{table}

The enclosing square of side $q=2R+1$ bounds the metric ball's $2R^2+2R+1$ sites. Using the Ising and Heisenberg term counts $L_{t}$, coefficient norms $\lambda_H$ and the rounded linear fits for $C_{\rm prep}$ of \cref{sec:numerics}, the query cost in logical Toffolis is
\begin{equation}
    C_H(R)=L_{t}(R)-2+2C_{\rm prep}(R)+2\lceil\log_2L_{t}(R)\rceil.
    \label{eq:cwalk}
\end{equation}
The count includes SELECT, two PREPs and reflection, with $q^2+60$ Hamiltonian work qubits. Set $b_N\coloneqq\lceil\log_2N\rceil$, coefficient precision $b_{\rm coef}=16$, $\eHS=10^{-3}$ and $\ell_N\coloneqq\lceil\log_2(N/\eHS)\rceil$. A separately generated branch has forward/backward query count and cost
\begin{equation}
    Q_j\coloneqq2\bigl[\lceil \lambda_H(R_j)|t_j|\rceil+\ell_N\bigr],\qquad
    W_j\coloneqq Q_jC_H(R_j).
    \label{eq:qj}
\end{equation}
This query-count model determines the cost estimate~\cite{LowChuang2019}; implementation accuracy is supplied independently through $\zeta_{\rm imp}$. With the accounting multiplier $A_{\rm norm}=\max\{1,\lceil\norm{\bs{c}^\star}_1\rceil\}$, the separate-branch cost is
\begin{equation}
C_{\rm BE}\coloneqq A_{\rm norm}\left[\sum_{j:t_j\ne0}z_jW_j+(N-2)+2b_N\right].
\label{eq:cbe}
\end{equation}
The fixed-$N$ overhead charges selection and reflection. This branch-by-branch cost model underlies \cref{thm:twostage}; reported resources use the compact ladder. Compared filters share $A_{\rm norm}$, which cancels in their ratios.

The ladder uses ordered positive weights $(w_1,\ldots,w_b)$ of minimum width $b=\lceil\log_2(K_{\rm act}+1)\rceil$, summing to $K_{\rm act}$ and representing every magnitude by subset sums. Evolving under $Z_s\otimes H_{a,\rho_c}$, with $Z_s$ acting on the sign qubit, runs either direction without extra Toffolis. Conjugating the coupling uses each weight's controlled call twice:
\begin{equation}
C^{\rm lad}_{\rm BE}\coloneqq A_{\rm norm}\left[2\sum_{c=1}^{b}
\bigl(\lceil\lambda_H(\rho_c)w_c\tau\rceil+\ell_N\bigr)C_H(\rho_c)+(N-2)+2b_N\right].
\label{eq:cladder}
\end{equation}
With uniform radii every digit uses the longest active branch's radius; grading uses \cref{prop:walkradii} at each cumulative weight. Exhaustive search selects the cheapest complete list and order at the tested design points ($K_{\rm act}\le20$). The ladder simulates a total time $2K_{\rm act}\tau$ and pays the surcharge $\ell_N$ once per controlled evolution, $2b$ times in all, retaining $N$ in the label overhead.

Preparing the coefficients by coherent alias sampling, and undoing the preparation, together cost~\cite{Harrigan2024}
\begin{equation}
C_T^{\rm alias}=2A_{\rm norm}(4N+8b_{\rm coef}+19b_N-8)
\label{eq:alias}
\end{equation}
logical $T$ gates, with $A_{\rm norm}(4Q_{\rm lad}+N)$ arbitrary-angle rotations, $Q_{\rm lad}\coloneqq2\sum_{c=1}^b[\lceil\lambda_H(\rho_c)w_c\tau\rceil+\ell_N]$, and $3b_N+2b_{\rm coef}+1$ alias work qubits beyond the largest patch width. Controlled and uncontrolled queries receive the same charge. The outer loop, channel splitting, resets, energy estimation, rotation synthesis, magic-state factories and routing are not costed here. The normalisation $\mathcal D[K/\norm{\bs{c}^\star}_1]=\mathcal D[K]/\norm{\bs{c}^\star}_1^2$ adds a factor $\norm{\bs{c}^\star}_1^2$ to the outer loop's run time, separately from these primitive costs.

With $\lambda_H,C_H=O(n_p)$ on $n_p$ patch sites and $\norm{\bs{c}^\star}_1=O(1)$, the full-support ladder costs $O(n_p^2T+n_pb\ell_N+N)$. Independent branches replace $T$ by $T^2/\tau$ and $b$ by $N$; support reduction replaces $T$ by the longest active time. Dense, light-cone and fixed-radius compilations take $n_p=n$, $\min\{n,V(2v_\LR T+\ell+2)\}$ and $V(R^*)$ at a fixed maximum radius $R^*$, respectively. Polynomial volume growth in dimension $D$ and $T=O(\Delta^{-1})$ give leading costs $n^2\Delta^{-1}$, $\Delta^{-(2D+1)}$ and $\Delta^{-1}$; grading changes constants. The Gevrey one-norm $\norm{\bs d}_1=\pi^{-1}\log(1/\Delta)+O(1)$ adds logarithmic factors through $A_{\rm norm}$ and, quadratically, through the outer loop's normalisation. Fixed-radius preparation requires small truncation error and a mixing estimate. Model-specific Lieb--Robinson bounds~\cite{WangHazzard2020} and geometric product formulas~\cite{ChildsSu2019} can change the costs. The velocity enters every radius linearly, while both the query count and the cost of each query grow with the patch area, so in two dimensions a tighter model-specific velocity can reduce the cost of a branch by up to the fourth power of the velocity ratio. The chain simulations use the generic irreducible-path velocity and the two-dimensional estimates simply set $v_\LR=1$, so substituting a model-specific velocity, or solving for it numerically, is the cheapest available improvement to the counts reported here.

The programs with whole-band (second-derivative) constraints use $N\in\{21,41\}$, $T\in\{2\pi,4\pi,6\pi,8\pi\}$, $f\in\{1,1.10,1.25\}$, band $[0.2,2]$, norm budget $f\norm{\bs d}_1$, reference cooling floor and $1200$ constraints. Displacement is capped at $0.8g_d/[2(1+f)\norm{\bs d}_1]$ with assumed margin $g_d=0.50$. Stage two uses $\xi=0.05$ and $\eta_{\sup}=\min\{\eta_{\rm budget},1.05\eta_{\rm spec}\}$, where $\eta_{\rm budget}$ is the displacement cap above and $\eta_{\rm spec}=\eta(\bs c_{\rm spec},\bs d)$, and minimises \cref{eq:cladder} over feasible symmetric truncations and weight orders. The coefficient programs are solved to a relative optimality gap of $10^{-8}$; ties are broken first by leakage and then by displacement, and the optimum need not be unique. Optimality applies to these specified programs and finite searches.

For $T=2\pi,4\pi,6\pi,8\pi$, respectively $5,4,3,3$ of the maximum radii $R^*\in\{16,32,64,96,128\}$ satisfy $\ell=R^*-1-2v_\LR T\ge0$. Combined with the two cost blocks and two supports, this gives $360$ cases. Full and shortest supports coincide for all $N=21$ cases and for $N=41$, $T\le4\pi$; counting those $144$ pairs once leaves $\NrowsTotal$ instances ($180$ full, $36$ shortened). The $\NrowsAlias$ cases at $N=21$, $T\ge6\pi$ cannot separate heating from cooling: the response has period $2\pi/\tau$, so $h_N(\pi/\tau)=h_N(-\pi/\tau)$, and $\pi/\tau$ lies inside the band. The remaining $\NrowsUsable$ instances use $C_4=\mu=v_\LR=1$ and require state error bounds of at most $2$. The archive records the checks; both cost blocks assume the mixing margin $g_d=0.50$.

In \cref{tab:resources}, $R^*=64$ covers the light cone, while $R^*=2$--$4$ give $25$--$81$ patch qubits and require direct convergence checks. Matched-leakage support reduction supplies a further factor $\SupportFactor{}$; global error targets can require radii growing as $\log m$.

The lattice rows fix $\Delta=0.2$; a fixed-coupling Hamiltonian rescaled to $[-1,1]$ has $\Delta=O(1/n)$. Physical estimates must scale gap, norms, velocity and time together (\cref{eq:energyscalingmain}). The assumed $\StepsAssumed$ applications reflect Ising-chain relaxation; Heisenberg takes several times longer. The number of calls made by a discrete-time simulation of the dynamics depends on the mixing, the normalisation and the channel splitting~\cite{DingChenLin2024}; for finite-range dissipation the simulation itself can be carried out with gate count near-linear in the space-time volume for sparsely located dissipation and $O((nt)^{4/3})$ up to polylogarithmic factors otherwise~\cite{Mizuta2026}. The logarithmic mixing time proved for weakly interacting systems~\cite{Zhan2025} would need the transfer argument discussed in \cref{sec:outro} before it applies to this generator. Coherent Hubbard and molecular-energy estimates~\cite{KanSymons2025,Lee2021,Georges2025}, and coherent estimates of other molecular observables~\cite{Steudtner2023}, concern different tasks and accounting.
\section{Numerical methods}
\label[appendix]{app:numerics}

\emph{Archive and models.} Release \texttt{numerical-2026-09-11}~\cite{FTDQENumerical2026} holds, for every table and figure, the coefficients, Hamiltonians, molecular coordinates, ladder digit assignments $\bs\chi(j)$, radii, settings, random-number-generator states and outputs. It keeps the feasibility checks, the optimality and support-minimality checks, the modelled costs and the dynamical validation separate; file hashes identify the data, while residuals, constraint checks and trajectory comparisons assess numerical correctness. Molecular inputs use PySCF/STO-3G, CASCI active spaces with frozen lowest orbitals where specified, and the neutral fixed-electron-number, $S_z=0$ spin-sector working space. Because the spin-resolved hoppings connect different total-spin sectors, triplet states set the gaps of the hydrogen systems. Hamiltonians are scaled to $[-1,1]$ on their working spaces; molecular Pauli counts and norms use Jordan--Wigner form, with Hubbard ordering $(i,\sigma)\mapsto2i+\sigma$.

\emph{Filter design.} Programs with second-derivative remainder constraints certify the bands in \cref{app:obs1,app:frontier}; per-model programs use sampled constraints and subsequent fine-grid evaluation. The Gevrey reference~\cite{DingChenLin2024} has $(\omega_{\rm w},\delta_{\rm w})=(2.5,0.5)$, $\tau=\pi/5$, requested $T_{\rm target}=5/\Delta$, nearest-integer $K=\operatorname{round}(T_{\rm target}/\tau)$ and realised $T=K\tau$. The sampled first stage uses cooling floor $h_{\eref,N}(-e)+2\times10^{-4}$, norm budget $\norm{\bs d}_1$ and displacement $\eta\le0.8g_d/(4\norm{\bs d}_1)$ at assumed $g_d=0.5$, on $1200$ energies, checked on $20{,}000$; above $1000$ nodes, constraint generation~\cite{Kelley1960} takes at most two rounds. The shortest feasible symmetric support is found by bisection (per-model) or enumeration (family), after which the displacement is minimised at fixed leakage. The per-model match programs drop the displacement cap, so their displacements exceed the budget of \cref{eq:linearcapmain} by factors of $2$ to $72$ and \cref{prop:transfer} does not cover them; their state accuracy rests on the simulations.

F$_2$ and H$_{10}$ are re-solved by adding the energies at which the constraints fail, with every constraint tightened by a margin of $5\times10^{-8}$. Acceptance requires the reference leakage cap and a $2\times10^{-4}$ cooling gain throughout the band, checked using $100{,}001$ samples, the local extrema of the response and the remainder $M_2h^2/8$ of \cref{eq:interpolation} on intervals of width $h$ (with $\bs c-\bs d$ in place of $\bs c$ for the gain), plus a $5\times10^{-11}$ floating-point allowance. Independent midpoint Taylor checks agree. The bounds are evaluated in ordinary floating point, without directed rounding. For H$_{10}$ the displacement is then minimised at fixed leakage; for F$_2$ it is not. For the next shorter symmetric support, the dual of the sampled program without the margin gives a lower bound on its leakage that lies above the cap, so the support is minimal.

\emph{Dynamics and diagnostics.} Jumps are evaluated in the eigenbasis of the Hamiltonian that generates them, $K_{kk'}=h_N(E'_k-E'_{k'})\bra{\psi'_k}A\ket{\psi'_{k'}}$, where $E'_k$ and $\ket{\psi'_k}$ are the eigenvalues and eigenvectors of that Hamiltonian; a patch Hamiltonian acts as the identity outside its patch.
Dynamics retains the full coherent Hamiltonian, unit dissipative rate and highest-energy initial state. Lattice dynamics are integrated numerically with time step $0.05$ (four/six sites) or $0.1$ (eight qubits), to times $60$ (four/six-site Ising), $40$ (eight-site Ising), $120$ (Heisenberg) or $100$ (Hubbard). For Hilbert-space dimension at most $64$, stationary states are computed by projecting onto the null space of the Lindbladian along its left null vectors. Molecular jumps are unit-norm spin-resolved hoppings; F$_2$ and Cl$_2$ are propagated to $t=160$ with a Krylov method that never forms the Lindbladian as a matrix; at that time $\tnorm{\calL[\rho]}$ and the change in population over the last five time units are both below $10^{-10}$, and both runs agree with step-by-step integration over the first ten time units. Lindbladian spectra are computed exactly; for Hubbard they are computed on the invariant ten-dimensional spin-singlet, pseudospin-singlet space, where a simple zero eigenvalue shows that each stationary state is unique. The Lindbladian gap bounds every mixing rate from above, and prefactors measured from sampled initial states bound the worst-case prefactor from below, so no mixing pair is established and the stationary-state bounds remain conditional. The archived nine-site test of graded digits uses $(J,h,\tau,K,R^*,\ell)=(1,0.7,0.2,5,3,0)$ and a central $X$ coupling; with weights $(2,1,2)$ at radii $(2,3,3)$, the error of branch $j=4$ is $1.24$ times that at uniform radius three. The random jump orders of \cref{tab:outerloop} come from NumPy's PCG64 generator initialised by \texttt{default\_rng(3)}, which draws the $60$ independent sites before the sweep permutations.

\end{document}